%% file: poe.tex
\PassOptionsToPackage{hyphens}{url}
\PassOptionsToPackage{table}{xcolor}
\documentclass[sigconf,edbt]{acmart-edbt2021}
\usepackage{microtype}

\def\changed{\color{black}}

\usepackage{moresize}
\input{header}

\setcopyright{rightsretained}

\acmDOI{}

\acmISBN{978-3-89318-084-4}

\acmConference[EDBT 2021]{24th International Conference on Extending Database Technology (EDBT)}{March 23-26, 2021}{Nicosia, Cyprus} 
\acmYear{2021}

\input{data}

\newcommand{\resultgraph}[5]{\begin{tikzpicture}[plot]
        \begin{axis}[xlabel={#3},ylabel={#4},title={#2},xtick={#5},legend to name={mainlegend},legend columns=-1]
            \addplot table[x={Nodes},y={POE}] {#1};
            \addplot table[x={Nodes},y={PBFT}] {#1};
            \addplot table[x={Nodes},y={SBFT}] {#1};
            \addplot table[x={Nodes},y={HS}] {#1};
            \addplot table[x={Nodes},y={ZYZ}] {#1};
            \addplot table[x={Nodes},y={POE}] {#1};
            \legend{\PoE{},\pbft{},\SBFT{},\hotstuff{},\ZZ{}};
        \end{axis}
    \end{tikzpicture}}

\newcommand{\latgraph}[2]{\begin{tikzpicture}[plot]
        \begin{axis}[xlabel={Latency (\si{\milli\second})},ylabel={Throughput (\si{\text{decisions}\per\second})},title={#2},legend to name={latlegend},legend columns=-1,xtick={10,20,40},ymin=0]
            \addplot table[x expr=\thisrow{LAT}/1000,y={POE}] {#1};
            \addplot table[x expr=\thisrow{LAT}/1000,y={PBFT}] {#1};
            \addplot table[x expr=\thisrow{LAT}/1000,y={HS}] {#1};
            \legend{\PoE{},\pbft{},\hotstuff{}};
        \end{axis}
    \end{tikzpicture}}

\newcommand{\latspecialgraph}[2]{\begin{tikzpicture}[plot]
        \begin{axis}[xlabel={Latency (\si{\milli\second})},ylabel={Throughput (\si{\text{decisions}\per\second})},title={#2},xtick={10,20,40},ymin=0]
            \addplot table[x expr=\thisrow{LAT}/1000,y={POEOO}] {#1};
            \addplot table[x expr=\thisrow{LAT}/1000,y={PBFTOO}] {#1};
            \legend{\PoE{}$*$,\pbft{}$*$};
        \end{axis}
    \end{tikzpicture}}

\newcommand{\axistput}{Throughput (\si{\text{txn}\per\second})}
\newcommand{\axislat}{Latency (\si{\second})}
\newcommand{\axisnodes}{Number of replicas ($\n$)}
\newcommand{\axisbatches}{Batch size}
\newcommand{\axisticksnodes}{4,16,32,64,91}
\newcommand{\axisticksbatches}{10,50,100,200,400}

\begin{document}

\title{Proof-of-Execution: Reaching Consensus through Fault-Tolerant Speculation}

\author{Suyash Gupta \quad{} Jelle Hellings \quad{} Sajjad Rahnama \quad{} Mohammad Sadoghi}%
\affiliation{%
\institution{Exploratory Systems Lab}
\institution{Department of Computer Science}
\institution{University of California, Davis}
}

\renewcommand{\shortauthors}{}

\begin{abstract}
\input{abstract}
\end{abstract}

\maketitle

\input{intro}
\input{back}
\input{dbft}
\input{appendix}
\input{impl}

\input{eval}

\input{simulation}
\input{related}
\input{concl}

\balance
\bibliographystyle{ACM-Reference-Format}
\bibliography{refined}

\end{document}

%% file: header.tex
\usepackage{amsthm,amsmath}
\usepackage[inline]{enumitem}
\usepackage{caption,subcaption}
\usepackage{balance}
\usepackage[binary-units]{siunitx}
\usepackage{graphicx}

\newcommand{\Replicas}{\mathfrak{R}}
\newcommand{\Faulty}{\mathcal{F}}
\newcommand{\n}{\mathbf{n}}
\newcommand{\f}{\mathbf{f}}
\newcommand{\nf}{\mathbf{nf}}
\newcommand{\CC}{\mathbf{c}}
\newcommand{\ID}[1]{\mathop{\textsf{id}}(#1)}
\newcommand{\Replica}[1][r]{\textsc{#1}}
\newcommand{\Primary}[1][p]{\textsc{#1}}
\newcommand{\Client}[1][c]{\MakeLowercase{#1}}

\newcommand{\Name}[1]{\texttt{#1}}
\newcommand{\PName}[1]{\textsc{#1}}
\newcommand{\BFT}{\PName{bft}}
\newcommand{\ZZ}{\PName{Zyzzyva}}
\newcommand{\PoE}{\PName{PoE}}
\newcommand{\PoEFull}{\textnormal{Proof-of-Execution}}
\newcommand{\pbft}{\PName{Pbft}}
\newcommand{\hotstuff}{\PName{HotStuff}}
\newcommand{\SBFT}{\PName{SBFT}}
\newcommand{\ExpoDB}{\PName{ResilientDB}}
\newcommand{\POW}{\PName{PoW}}

\newcommand{\Transaction}[1][t]{\MakeUppercase{#1}}
\newcommand{\MName}[1]{\textsc{#1}}
\newcommand{\Message}[2]{\textsc{#1}(#2)}
\newcommand{\SignMessage}[2]{\langle#1\rangle_{#2}}
\newcommand{\SignShare}[2]{s\langle#1\rangle_{#2}}
\newcommand{\Hash}[1]{\texttt{D}(#1)}

\newcommand{\ViewCommitted}[4]{\texttt{VCommit}_{#1}(#2, #3, #4)}
\newcommand{\Executed}[4]{\texttt{Execute}_{#1}(#2, #3, #4)}

\newcommand{\abs}[1]{\lvert #1 \rvert}
\newcommand{\union}{\cup}
\newcommand{\intersect}{\cap}
\newcommand{\difference}{\setminus}
\newcommand{\BigO}[1]{\mathcal{O}(#1)}

\usepackage{algorithmic}
\newcommand{\GETS}{:=}
\newenvironment{myprotocol}{
    \hrule
    \smallskip
    \footnotesize
    \algsetup{linenosize=\footnotesize}
    \begin{algorithmic}[1]
        \newcommand{\SPACE}{\item[]}
        \newcommand{\TITLE}[2]{\item[] \textbf{\underline{##1}} (##2) \textbf{:}\\[2pt]}
        \makeatletter
            \newcommand{\EVENT}[1]{\STATE \textbf{event} ##1 \textbf{do}}
            \newcommand{\ENDEVENT}{ \STATE \textbf{end event}}
        \makeatother
}{
    \end{algorithmic}
    \smallskip
    \hrule
}

\usepackage{tikz,pgfplots,pgfplotstable}
\usetikzlibrary{arrows.meta,decorations.pathmorphing,decorations.pathreplacing}
\tikzset{
    >=Stealth,
    dot/.style={circle,scale=0.35,draw=black,fill=black},
    node_text/.append style={font=\strut\bfseries},
    label/.append style={font=\strut\large},
    plot/.append style={baseline,scale=0.475}
}

\pgfplotscreateplotcyclelist{mycyclelist}{
    red,every mark/.append style={solid,fill=\pgfplotsmarklistfill},mark=*\\
    blue,every mark/.append style={solid,fill=\pgfplotsmarklistfill},mark=square*\\
    brown,every mark/.append style={solid,fill=\pgfplotsmarklistfill},mark=triangle*\\
    purple!80!black,every mark/.append style={solid},mark=star\\
    green!80!black,every mark/.append style={solid,fill=\pgfplotsmarklistfill},mark=diamond*\\
}

\pgfplotsset{
    tick label style={font=\huge},
    legend style={font=\huge,cells={anchor=west}},
    title style={font=\huge},
    label style={font=\huge},
    width=265pt,
    height=185pt,
    every axis/.append style={
        ylabel near ticks,
        xlabel near ticks,
        mark size=2pt,
        cycle list name=mycyclelist,
        font=\huge,
        legend pos=north east
    },
    barstyle/.append style={
        width=5.6cm,
        ybar,
        bar width={1cm},
        enlarge x limits=0.4,
        enlarge y limits={upper=0.025},
        ymin=0,
        xtick=data,
    }
}

\newcommand{\revised}[1]{{\color{black} #1}}

%% file: data.tex
\pgfplotstableread{
method	measure
ne  551054
e   217345
}\dataTPUTfnEXEC

\pgfplotstableread{
method	measure
ne  0.244
e   0.639
}\dataLATfnEXEC

\pgfplotstableread{
method	measure
ns  156277
ed  21855
cmac    104354
}\dataTPUTfnSIGNS

\pgfplotstableread{
method	measure
ns  0.501
ed  3.644
cmac    0.751
}\dataLATfnSIGNS

\pgfplotstableread{
Nodes	POE	SBFT	PBFT	HS	ZYZ
4	179165	104080	178514	7207	2833
16	157170	98649	121185	7150	3330
32	122312	86746	85657	7000	4756
64	63886	46930	52283	6500	7849
91	41356	33602	35593	7207	7919	
}\dataTputNodes

\pgfplotstableread{
Nodes	POE	SBFT	PBFT	HS	ZYZ
10	22477	10820	13523	800	70
50	78337	47354	52676	3800	2212
100	122312	86746	85657	7207	2264
200	167888	103784	119019	9000	6859
400	195817	103656	156548	12000	11129
}\dataTputBatch

\pgfplotstableread{
Nodes	POE	SBFT	PBFT	HS	ZYZ
4	224198	104283	275848	8960	2264
16	195027	104047	166922	9100	2264
32	193236	88411	109430	9169	2264
64	118866	51470	64658	9126	2264
91	73684	37009	44436	9074	2264
}\dataTputNodesZeroPL

\pgfplotstableread{
Nodes	POE	SBFT	PBFT	HS	ZYZ
4	247333	212921	244794	8964	254825
16	231626	166623	150672	9160	245556
32	197411	121767	106652	9163	220545
64	121449	70723	65530	9186	125786
91	93808	54703	46164	9004	87311
}\dataTputNodesZeroPLFF

\pgfplotstableread{
Nodes	POE	SBFT	PBFT	HS	ZYZ
4	201691	158571	194716	9072	230990
16	159874	132600	149677	9072	199895
32	128835	111412	95217	9296	154506
64	67359	60397	55903	9096	80997
91	47136	47687	38455	8952	59203
}\dataTputNodesFF

\pgfplotstableread{
Nodes	POE	SBFT	PBFT	HS	ZYZ
4	3.768	3.658	3.968	9.140	3.829	
16	3.353	3.622	3.455	9.065	3.626	
32	3.006	4.108	3.166	9.113	2.684	
64	2.864	3.793	2.744	8.988	2.490	
91	2.876	3.689	2.489	8.819	1.948	
}\dataTputNodesChain

\pgfplotstableread{
Nodes	POE	SBFT	PBFT	HS	ZYZ
4	1.687	3.018	1.672	0.032	4.078
16	1.99	3.191	2.495	0.033	3.806
32	2.549	3.627	3.512	0.032	3.522
64	4.792	6.556	5.754	0.033	3.686
91	7.388	9.338	8.458	0.034	3.816
}\dataLatNodes

\pgfplotstableread{
Nodes	POE	SBFT	PBFT	HS	ZYZ
10	13.967	30.464	21.737	0.060	3.88
50	3.995	6.597	5.666	0.040	3.684
100	2.549	3.627	3.512	0.032	4.004
200	1.844	3.028	2.522	0.030	3.499
400	1.548	3.027	1.952	0.031	3.416
}\dataLatBatch

\pgfplotstableread{
Nodes	POE	SBFT	PBFT	HS	ZYZ
4	1.382	3.002	1.046	0.031	4.078
16	1.595	3.021	1.867	0.032	3.806
32	1.588	3.563	2.82	0.034	3.522
64	2.717	6.038	4.796	0.033	3.586
91	4.209	8.197	6.981	0.039	3.616
}\dataLatNodesZeroPL

\pgfplotstableread{
Nodes	POE	SBFT	PBFT	HS	ZYZ
4	0.311	0.365	0.31	0.034	0.618
16	1.333	1.905	2.015	0.034	1.24
32	1.586	2.862	2.934	0.034	1.416
64	2.605	4.786	4.794	0.034	2.546
91	3.318	5.675	6.704	0.035	3.586
}\dataLatNodesZeroPLFF
    
\pgfplotstableread{
Nodes	POE	SBFT	PBFT	HS	ZYZ
4	1.597	2.005	1.56	0.033	1.308
16	1.962	2.632	2.046	0.033	1.499
32	2.481	2.943	3.304	0.033	2.498
64	4.644	5.119	5.576	0.035	4.925
91	6.582	6.738	7.801	0.036	6.543
}\dataLatNodesFF
    
\pgfplotstableread{
Nodes	POE	SBFT	PBFT	HS	ZYZ
4	0.018	0.018	0.017	0.033	0.018
16	0.021	0.019	0.02	0.034	0.02
32	0.024	0.016	0.023	0.034	0.028
64	0.027	0.017	0.027	0.035	0.032
91	0.028	0.018	0.032	0.036	0.041
}\dataLatNodesChain

\pgfplotstableread{
LAT	PBFTOO	POEOO	PBFT	POE	HS	
10000	4629.62963	5376.344086	16.51745895	16.54204989	32.0410125	
20000	2487.562189	2688.172043	9.361893349	9.723086497	16.02050625		
40000	1694.915254	1785.714286	6.738544474	6.595522959	10.68010936	
}\dataSimLatencyTPutFOUR

\pgfplotstableread{
LAT	PBFTOO	POEOO	PBFT	POE	HS	
10000	2392.344498	4854.368932	18.45972089	18.37154615	32.04511953
20000	2092.050209	2525.252525	9.635025244	10.22536709	16.02101958
40000	1449.275362	1754.385965	6.659740537	6.786654722	10.68056564
}\dataSimLatencyTPutSIXTEEN

\pgfplotstableread{
LAT	PBFTOO	POEOO	PBFT	POE	HS	
10000	778.8161994	4464.285714	19.85387548	18.69927821	32.01844262
20000	759.8784195	2427.184466	9.29074456	9.624268556	16.02204634
40000	740.7407407	1661.129568	6.48281406	6.469309595	10.67760053
}\dataSimLatencyTPutONETWOEIGHT

\pgfplotstableread{
npr	nt	POE	PBFT
4	20	92000	68000
7	35	86000	61000
13	65	64000	49000
19	95	43000	33000
28	140	29000	22000
}\dataTPutWAN
            
\pgfplotstableread{
npr	nt	POE	PBFT
4	20	8.297 	11.632
7	35	9.610 	13.340
13	65	13.931	15.696
19	95	20.810	22.693
28	140	29.306	33.086
}\dataLatWAN

%% file: abstract.tex
Multi-party data management and blockchain systems require data sharing among participants. To provide resilient and consistent data sharing, transactions engines rely on Byzantine Fault-Tolerant consensus (\BFT{}), which enables operations during failures and malicious behavior. Unfortunately, existing \BFT{} protocols are unsuitable for high-throughput applications due to their high computational costs, high communication costs, high client latencies, and/or reliance on twin-paths and non-faulty clients.

In this paper, we present the \emph{Proof-of-Execution consensus protocol} (\PoE{}) that alleviates these challenges. At the core of \PoE{} are \emph{out-of-order} processing and \emph{speculative execution}, which allow \PoE{} to execute transactions before consensus is reached among the replicas. With these techniques, \PoE{} manages to reduce the costs of \BFT{} in normal cases, while guaranteeing reliable consensus for clients in all cases.  We envision the use of \PoE{} in high-throughput multi-party data-management and blockchain systems. To validate this vision, we implement \PoE{} in our efficient \ExpoDB{} fabric and extensively evaluate \PoE{} against several state-of-the-art \BFT{} protocols. {\changed{} Our evaluation showcases that \PoE{} achieves  up-to-$80\%$ higher throughputs than existing \BFT{} protocols in the presence of failures}.

%% file: intro.tex
\section{Introduction}  
\label{s:intro}

In \emph{federate data management} a single common database is managed by many independent stakeholders (e.g., an industry consortium). In doing so, federated data management can ease data sharing and improve data quality~\cite{dqbottom,dqbook,impactent}. At the core of federated data management is \emph{reaching agreement} on any updates on the common database in an efficient manner, this to enable fast query processing, data retrieval, and data modifications.

One can achieve federated data management by \emph{replicating} the common database among all participant, this by replicating the sequence of transactions that affect the database to all stakeholders. One can do so using commit protocols designed for distributed databases such as two-phase~\cite{2pc} and three-phase commit~\cite{3pc}, or by using crash-resilient replication protocols such as Paxos~\cite{paxos} and Raft~\cite{raft}.

These solutions are error-prone in a federated \emph{decentralized} environment in which each stakeholder manages its own replicas and replicas of each stakeholder can fail (e.g., due to software, hardware, or network failure) or act malicious: \revised{commit protocols and replication protocols can only deal with crashes}. Consequently, recent federated designs propose the usage of Byzantine Fault-Tolerant (\BFT{}) consensus protocols. \BFT{} consensus aims at \emph{ordering client requests among a set of replicas, some of which could be Byzantine, such that all non-faulty replicas reach agreement on a common order for these requests}~\cite{pbftj,zyzzyva,hotstuff,sbft,geobft}. Furthermore, \BFT{} consensus comes with the added benefit of \emph{democracy}, as \BFT{} consensus gives all replicas an equal vote in all agreement decisions, while the resilience of \BFT{} can aid in dealing with the billions of dollars losses associated with prevalent attacks on data management systems~\cite{ecodam}.

Akin to commit protocols, the majority of \BFT{} consensus protocols use a \emph{primary-backup model} in which one replica is designated \emph{the primary} that coordinates agreement, while the remaining replicas act as backups and follow the protocol~\cite{pdbook}. This primary-backup \BFT{} consensus was first popularized by the influential \pbft{} consensus protocol of Castro and Liskov~\cite{pbftj}. The design of \pbft{} requires at least $3\f+1$ replicas to deal with up-to-$\f$ malicious replicas and operates in \emph{three communication phases}, two of which necessitate quadratic communication complexity. As such, \pbft{} is considered costly when compared to commit or replication protocols, which has negatively impacted the usage of \BFT{} consensus in large-scale data management systems~\cite{improv-bft}.

The recent interest in blockchain technology has revived interest in \BFT{} consensus, has led to several new resilient data management systems (e.g.,~\cite{caper,blockchaindb,geobft,blockplane}), and has led to the development of new \BFT{} consensus protocols that promise efficiency at the cost of flexibility (e.g.,~\cite{zyzzyva,sbft,hotstuff,rcc}). Despite the existence of these modern \BFT{} consensus protocols, the majority of \BFT{}-fueled systems~\cite{geobft,caper,blockchaindb,blockplane} still employ the classical time-tested, flexible, and safe design of \pbft{}, however.

In this paper, we explore different design principles that can enable implementing a scalable and reliable 
agreement protocol that shields against malicious attacks.
We use these design principles to introduce \PoEFull{} (\PoE{}), a novel \BFT{} protocol  that achieves 
resilient agreement in just three linear phases.
To concoct \PoE's scalable and resilient design, we start with \pbft{} and successively add {\em four} design elements:

\begin{enumerate}[wide,nosep,label=(I\arabic*),ref={I\arabic*}]

\item \label{req:NDSE} {\bf Non-Divergent Speculative Execution.} 
In \pbft{}, when the primary replica receives a client request, it forwards that request to the backups.
Each backup on receiving a request from the primary agrees to support by broadcasting a \MName{prepare} message. 
When a replica receives \MName{prepare} message from the majority of other replicas, it marks itself as {\em prepared} 
and broadcasts a \MName{commit} message. 
Each replica that has prepared, and receives \MName{commit} messages from a majority of other replicas, executes the request.

Evidently, \pbft{} requires two phases of {\em all-to-all} communication.
Our first ingredient towards faster consensus is speculative execution. 
In \pbft{} terminology, \PoE{} replicas execute requests after they get {\em prepared}, that is, they 
do not broadcast \MName{commit} messages.
This speculative execution is non-divergent as each replica has a partial guarantee--it has prepared--prior to execution.

\item \label{req:SR} {\bf Safe Rollbacks and Robustness under Failures.}
Due to speculative execution, a malicious primary in \PoE{} can ensure that only a
subset of replicas prepare and execute a request. 
Hence, a client may or may not receive a sufficient number of matching responses. 
\PoE{} ensures that if a client receives a \emph{full proof-of-execution}, consisting of responses from a majority of the non-faulty replicas, then such a request persists in time.
Otherwise, \PoE{} permits replicas to {\em rollback} their state if necessary.
This proof-of-execution is the cornerstone of the correctness of \PoE{}.

\item \label{req:SSLC} {\bf Agnostic Signatures and Linear Communication.}
\BFT{} protocols are run among distrusting parties. To provide security,
these protocols employ cryptographic primitives for signing the messages and generating 
message digests.
Prior works have shown that the choice of cryptographic signature scheme can impact the 
performance of the underlying system~\cite{pbftj,resilientdb}.
Hence, we allow replicas to either employ {\em message authentication codes} (\Name{MAC}s) or 
{\em threshold signatures} (\Name{TS}s) for signing~\cite{cryptobook}.
When few replicas are participating in consensus (up to $16$), then a single phase of 
all-to-all communication is inexpensive and using \Name{MAC}s for such setups can make computations cheap.
For larger setups, we employ \Name{TS}s to achieve linear communication complexity.
\Name{TS}s permit us to split a phase of all-to-all communication into {\em two linear phases}~\cite{hotstuff,sbft}.

\item \label{req:ARA} {\bf Avoid Response Aggregation.} 
\SBFT{}~\cite{sbft}, a recently-proposed \BFT{} protocol,
suggests the use of a single replica (designated as the {\em executor}) to act as a response aggregator.
In specific, all replicas execute each client request and send their response to the executor.
It is the duty of the executor to reply to the client and {\em send a proof} that a majority of the replicas 
not only executed this request, but also outputted the same result.
In \PoE{}, we avoid this additional communication between the replicas
by allowing each replica to respond directly to the client.

\end{enumerate}

In specific, we make the following contributions:
\begin{enumerate}
    	\item \revised{We introduce \PoE{}, a novel Byzantine fault-tolerant consensus protocol that uses \emph{speculative execution} to reach agreement among replicas.}
    	\item To guarantee failure recovery in the presence of speculative execution and Byzantine behavior, we introduce a novel view-change protocol that can rollback requests.
        \item \PoE{} supports batching, out-of-order processing, and is signature-scheme agnostic and can be made to employ either \Name{MAC}s or threshold signatures.
    	\item \PoE{} does not rely on non-faulty replicas, clients, or trusted hardware to achieve safe and efficient consensus.
    	\item To validate our vision of using \PoE{} in resilient federated data management systems, we implement \PoE{} and four other \BFT{} protocols (\ZZ{}, \pbft{}, \SBFT{}, and \hotstuff{}) in our efficient \ExpoDB{}\footnote{\ExpoDB{} is open-sourced at https://github.com/resilientdb.} fabric~\cite{geobft,resilientdb,bftbook,vldb-demo,bft-tutorial-vldb20,bft-tutorial-debs20,bft-tutorial-middleware19}.
	\item We extensively evaluate \PoE{} against these protocols on a Google Cloud deployment consisting of $91$ replicas and $\SI{320}{\kilo{}}$ clients under
	(i) no failure,
	(ii) backup failure, 
	(iii) primary failure, 
	(iv) batching of requests,
	(v) zero payload, and
	(vi) scaling the number of replicas.
	Further, to prove the correctness of our results, we also stress test \PoE{} and other protocols in a simulated environment.
    	Our results show that \PoE{} can achieve up to $80\%$ more throughput than existing \BFT{} protocols \revised{in the presence of failures}.

\end{enumerate}

To the best of our knowledge, \PoE{} is the first protocol that achieves consensus in \emph{only two phases} while being able to deal with Byzantine failures and without relying on trusted clients (e.g., \ZZ{}~\cite{zyzzyva}) or on trusted hardware (e.g., \PName{MinBFT}~\cite{minbft}). Hence, \PoE{} can serve as a drop-in replacement of \pbft{} to improve scalability and performance in permissioned blockchain fabrics such as our \ExpoDB{} fabric~\cite{bc-processing,geobft,resilientdb,rcc,bftbook}, MultiChain~\cite{multichain}, and Hyperledger Fabric~\cite{hyperledger-fabric}; in multi-primary meta-protocols such as RCC~\cite{rcc,multibft-disc}; and in sharding protocols such as AHL~\cite{ahl}.

%% file: back.tex
\section{Analysis of Design Principles}\label{sec:anal}
{\changed{}
To arrive at an optimal design for \PoE{}, we studied practices followed by state-of-the-art distributed data management systems and applied their principles to the design of \PoE{} where possible. In Figure~\ref{fig:compare}, we present a comparison of \PoE{} against \emph{four} well-known resilient consensus protocols.

\begin{figure}[t!]
    \centering
    \newcommand{\Bad}{\cellcolor{red!20}}
    \newcommand{\Good}{\cellcolor{green!30}}
    \scalebox{0.7}{
        \begin{tabular}{l||c@{\ \ }c@{\ \ }c@{\ \ }c@{\ \ }c@{\ \ }c}
        Protocol    			& Phases 	& Messages 			& Resilience 	& Requirements\\
        \hline
        \ZZ{}       			& 1      	& $\BigO{\n}$			&\Bad{}0 	&\Bad{}Reliable clients and unsafe \\
        \Good{}\PoE{} (our paper)     	&\Good{}3      	&\Good{}$ \Good{}\BigO{3\n}$	&\Good{}$\f$ 	&\Good{}Sign.\ agnostic\\
        \pbft{}     			&\Bad{}3      	&\Bad{}$\BigO{\n + 2\n^2}$	& $\f$ 		& \\
        \hotstuff{}				&\Bad{}8	& $\BigO{8\n}$  		& $\f$ 		&\Bad{}Sequential Consensus\\
        \SBFT{}				&\Bad{}5  	& $\BigO{5\n}$   		& $0$ 		&\Bad{}Optimistic path\\
        \end{tabular}
    }
    \caption{Comparison of \BFT{} consensus protocols in a system with $\n$ replicas of which $\f$ are faulty. The costs given are for the normal-case behavior.}
    \label{fig:compare}
\end{figure}

To illustrate the merits of \PoE's design, we first briefly look at \pbft{}. The last phase of \pbft{} ensures that non-faulty replicas only execute requests and inform clients when there is a guarantee that such a transaction will be recovered after any failures. Hence, clients need to wait for only $\f+1$ identical responses, of which at-least one is from a non-faulty replica, to ensure \emph{guaranteed execution}. By eliminating this last phase, replicas speculatively execute requests before obtaining recovery guarantees. This impacts \pbft{}-style consensus in two ways:
\begin{enumerate}
\item First, clients need a way to determine \emph{proof-of-execution} after which they have a guarantee that their requests are executed and maintained by the system. We shall show that such a proof-of-execution can be obtained using $\nf \geq 2\f+1$ identical responses (instead of $\f+1$ responses).
\item Second, as requests are executed before they are guaranteed, replicas need to be able to rollback requests that are dropped during periods of recovery.
\end{enumerate}
\PoE's speculative execution guarantees that requests with a proof-of-execution will never rollback and that only a single request can obtain a proof-of-execution per round. Hence, speculative execution provides the same strong consistency (safety) of \pbft{} in all cases, this at much lower cost under normal operations. Furthermore, we show that speculative execution is fully compatible with other scalable design principles applied to \pbft{}, e.g., batching and out-of-order processing to maximize throughput, even with high message delays. 

\textbf{Out-of-order execution.}
Typical \BFT{} systems follow the \emph{order-execute} model: first replicas agree on a unique order of the client request, and only then they execute the requests in order~\cite{pbftj,zyzzyva,hotstuff,sbft,geobft}. Unfortunately, this prevents these systems from providing any support for concurrent execution. A few \BFT{} systems suggest executing prior to ordering, but even such systems need to re-verify their results prior to committing changes~\cite{hyperledger-fabric,eve}. Our \PoE{} protocol lies between these two extremes: the replicas speculatively execute using only partial ordering guarantees. By doing so, \PoE{} can eliminate communication costs and minimize latencies of typical \BFT{} systems, this without needing to re-verify results in the normal case.

\textbf{Out-of-order processing.}
Although \BFT{} consensus typically executes requests in-order, this does not imply they need to process proposals to order requests sequentially. To maximize throughput, \pbft{} and other primary-backup protocols support \emph{out-of-order processing} in which all available bandwidth of the primary is used to continuously propose requests (even when previous proposals are still being processed by the system). By doing so, out-of-order processing can eliminate the impact of high message delays. To provide out-of-order processing, all replicas will process any request proposed as the $k$-th request whenever $k$ is within some \emph{active window} bounded by a \emph{low-watermark} and \emph{high-watermark}~\cite{pbftj}. These watermarks are increased as the system progresses. The size of this active window is---in practice---only limited by the memory resources available to replicas. As out-of-order processing is an essential technique to deliver high throughputs in environments with high message delays, we have included out-of-order processing in the design of \PoE{}.

\textbf{Twin-path consensus.}
Speculative execution employed by \PoE{} is different that the \emph{twin-path model} utilized by \ZZ{}~\cite{zyzzyva} and \SBFT{}~\cite{sbft}. These twin-path protocols have an optimistic \emph{fast} path that works only if none of the replicas are \emph{faulty} and require aid to determine whether these optimistic condition hold.

In the fast path of \ZZ{}, primaries propose requests, and backups directly execute such proposals and inform the client (without further coordination). The client waits for responses from all $\n$ replicas before marking the request executed. When the client does not receive $\n$ responses, it \emph{timeouts} and sends a message to all replicas, after which the replicas perform an expensive client-dependent \emph{slow-path} recovery process (which is prone to errors when communication is unreliable~\cite{zyzzyva-unsafe}). 

The fast path of \SBFT{} can deal with up to $\CC$ crash-failures using $3\f + 2\CC + 1$ replicas and uses threshold signatures to make communication linear. The fast path of \SBFT{} requires a reliable collector and executor to aggregate messages and to send only \emph{a single} (instead of at-least-$\f+1$) response to the client. Due to aggregating execution, the fast path of \SBFT{} still performs four rounds of communication before the client gets a response, whereas \PoE{} only uses two rounds of communication (or three when \PoE{} uses threshold signatures). If the fast path \emph{timeouts} (e.g., the collector or executor fails), then \SBFT{} falls back to a threshold-version of \pbft{} that takes an additional round before the client gets a response. 
Twin-path consensus is in sharp contrast with the design of \PoE{}, which does not need outside aid (reliable clients, collectors, or executors), and can operate optimally even while dealing with replica failures.

\textbf{Primary rotation.}
To minimize the influence of any single replica on \BFT{} consensus, \hotstuff{} opts to replace the primary after every consensus decision. To efficiently do so, \hotstuff{} uses an extra communication phase (as compared to \pbft{}), which minimizes the cost of primary replacement. Furthermore, \hotstuff{} uses threshold signatures to make its communication linear (resulting in eight communication phases before a client gets responses). The event-based version of \hotstuff{} can overlap phases of consecutive rounds, thereby assuring that consensus of a client request starts in every one-to-all-to-one communication phase. Unfortunately, the primary replacements require that all consensus rounds are performed in a strictly \emph{sequential} manner, eliminating any possibility of \emph{out-of-order processing}. 
}

%% file: dbft.tex
\section{Proof-of-Execution}
\label{s:dbft}

In our \emph{Proof-of-Execution consensus protocol} (\PoE{}), the primary replica is responsible for proposing transactions requested by clients to all backup replicas.
Each backup replica \emph{speculatively} executes these transactions with the belief that the primary  is behaving correctly.
Speculative execution expedites processing of transactions in all cases. Finally, when malicious behavior is detected, replicas can recover by \emph{rolling back transactions}, which ensures correctness without depending on any twin-path model. 

\input{prelim}

\subsection{The Normal-Case Algorithm of \PoE{}}\label{ss:poe_normal}

\begin{figure}[t!]
    \centering
    \begin{tikzpicture}[yscale=0.3,xscale=0.95]
        \draw[thick,draw=black!75] (0.75, 11.5) edge[green!50!black!90] ++(8, 0)
                                   (0.75,   7) edge[red!50!black!90] ++(8, 0)
                                   (0.75,   8) edge ++(8, 0)
                                   (0.75,   9) edge ++(8, 0)
                                   (0.75,   10) edge[blue!50!black!90] ++(8, 0);

        \draw[thin,draw=black!75] (1,   7) edge ++(0, 4.5)
                                  (2.5, 7) edge ++(0, 4.5)
				  (4, 7) edge ++(0, 4.5)
                                  (7, 7) edge ++(0, 4.5)
                                  (8.5,   7) edge ++(0, 4.5);

        \node[left] at (0.8, 7) {$\Replica[b]$};
        \node[left] at (0.8, 8) {$\Replica_2$};
        \node[left] at (0.8, 9) {$\Replica_1$};
        \node[left] at (0.8, 10) {$\Primary$};
        \node[left] at (0.8, 11.5) {$\Client$};

        \path[->] (1, 11.5) edge node[above] {$\Transaction$} (2.5, 10)
                  (2.5, 10) edge (4, 9)
                            edge (4, 8)
                            edge (4, 7)
                           
                  (4, 9) edge (7, 10)
			 edge (7, 8)
			 edge (7, 7)

		  (4, 8) edge (7, 10)
			 edge (7, 9)
			 edge (7, 7)

                  (7, 8) edge (8.5, 11.5)
                  (7, 9) edge (8.5, 11.5)
                  (7, 10) edge (8.5, 11.5)
                           ;

        \node[below] at (4.5,6.85) {\footnotesize {(a) \PoE{} using \Name{MAC}s}};

        \draw[thick,draw=black!75] (0.75, 4.5) edge[green!50!black!90] ++(8, 0)
                                   (0.75,   0) edge[red!50!black!90] ++(8, 0)
                                   (0.75,   1) edge ++(8, 0)
                                   (0.75,   2) edge ++(8, 0)
                                   (0.75,   3) edge[blue!50!black!90] ++(8, 0);

        \draw[thin,draw=black!75] (1,   0) edge ++(0, 4.5)
                                  (2.5, 0) edge ++(0, 4.5)
				  (4, 0) edge ++(0, 4.5)
                                  (5.5,   0) edge ++(0, 4.5)
                                  (7, 0) edge ++(0, 4.5)
                                  (8.5,   0) edge ++(0, 4.5);

        \node[left] at (0.8, 0) {$\Replica[b]$};
        \node[left] at (0.8, 1) {$\Replica_2$};
        \node[left] at (0.8, 2) {$\Replica_1$};
        \node[left] at (0.8, 3) {$\Primary$};
        \node[left] at (0.8, 4.5) {$\Client$};

        \path[->] (1, 4.5) edge node[above] {$\Transaction$} (2.5, 3)
                  (2.5, 3) edge (4, 2)
                           edge (4, 1)
                           edge (4, 0)
                           
                  (4, 2) edge (5.5, 3)
                  (4, 1) edge (5.5, 3)

		  (5.5, 3) edge (7, 2)
                           edge (7, 1)
                           edge (7, 0)

                  (7, 1) edge (8.5, 4.5)
                  (7, 2) edge (8.5, 4.5)
                  (7, 3) edge (8.5, 4.5)
                           ;

        \node[below] at (3.2, 0) {\footnotesize \strut\MName{propose}};
        \node[below] at (4.7, 0) {\footnotesize \strut\MName{support}};
        \node[below] at (6.2, 0) {\footnotesize \strut\MName{certify}};
        \node[below] at (7.7, 0) {\footnotesize \strut\MName{inform}};
        \node[below] at (4.5,-1.5) {\footnotesize {(b) \PoE{} using \Name{TS}s.}};
    \end{tikzpicture}
    \caption{Normal-case algorithm of \PoE{}: Client $\Client$ sends its request containing transaction $\Transaction$ to the primary $\Primary$, which proposes this request to all replicas.  
Although replica $\Replica[b]$ is Byzantine, it fails to affect \PoE{}.}\label{fig:poe}
\end{figure}
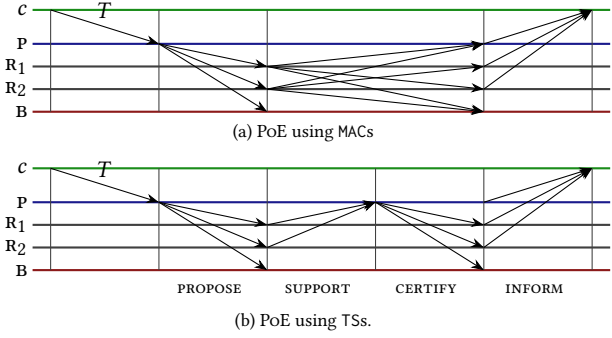

\PoE{} operates in \emph{views} $v = 0, 1, \dots$. 
In view $v$, replica $\Replica$ with $\ID{\Replica} = v \bmod \n$ 
is elected as the primary. The design of \PoE{} relies on authenticated communication, which can be provided using \Name{MAC}s or \Name{TS}s. In Figure~\ref{fig:poe}, we sketch the normal-case working of \PoE{} for both cases. For the sake of brevity, we will describe \PoE{} built on top of \Name{TS}s,  which results in a protocol with low---\emph{linear}---message complexity in the normal case. The full pseudo-code for this algorithm can be found in Figure~\ref{fig:pa}. In Section~\ref{app:mac}, we detail the minimal changes to \PoE{} necessary when switching to \Name{MAC}s.

Consider a view $v$ with primary $\Primary$. To request execution of transaction $\Transaction$, a client $\Client$ signs transaction $\Transaction$ and sends the signed transaction $\SignMessage{\Transaction}{\Client}$ to $\Primary$. The usage of signatures assures that malicious primaries cannot forge transactions. To initiate replication and execution of $\Transaction$ as the $k$-th transaction, the primary proposes $\Transaction$ to all replicas via a \MName{propose} message.

After the $i$-th replica $\Replica$ receives a \MName{propose} message $m$ from $\Primary$, it checks whether at least $\nf$ other replicas received the same proposal $m$ from primary $\Primary$. This check assures $\Replica$ that at least $\nf - \f$ non-faulty replicas received the same proposal, which will play a central role in achieving speculative non-divergence.  To perform this check, each replica \emph{supports} the first proposal $m$ it receives from the primary by computing a \emph{signature share} $\SignShare{m}{i}$ and sending a \MName{support} message containing this share to the primary. 

The primary $\Primary$ waits for \MName{support} messages with valid signature shares from  $\nf$ distinct replicas, which can then be aggregated into a single signature $\SignMessage{m}{}$. After generating such a signature, the primary broadcasts this signature to all replicas via a \MName{certify} message.

After a replica $\Replica$ receives a valid \MName{certify} message, it \emph{view-commits} to $\Transaction$ as the $k$-th transaction in view $v$. The replica logs this view-commit decision as {\changed $\ViewCommitted{\Replica}{\SignMessage{\Transaction}{\Client}}{v}{k}$}. After $\Replica$ view-commits to $\Transaction$, $\Replica$ schedules $\Transaction$ for 
speculative execution as the $k$-th transaction of view $v$. Consequently, $\Transaction$ will be executed by $\Replica$ after all preceding transactions are executed.  We write {\changed $\Executed{\Replica}{\SignMessage{\Transaction}{\Client}}{v}{k}$} to log this execution. 

After execution, $\Replica$ informs the client of the order of execution and 
of execution result $r$ (if any) via a message \MName{inform}. {\changed In turn, client $\Client$ will wait for a \emph{proof-of-execution} for the transaction $\Transaction$ it requested, which consists of identical \MName{inform} messages from $\nf$ distinct replicas. This proof-of-execution guarantees that at least $\nf - \f \geq \f+1$ non-faulty replicas executed $\Transaction$ as the $k$-th transaction and in Section~\ref{subsec:vc}, we will see that such transactions are always preserved by \PoE{} when recovering from failures.}

If client $\Client$ does not know the current primary or does not get any timely response for its requests, then it can broadcast its request $\SignMessage{\Transaction}{\Client}$ to all replicas. The non-faulty replicas will then forward this request to the current primary (if $\Transaction$ is not yet executed) and ensure that the primary initiates successful proposal of this request in a timely manner.


\begin{figure}[t]
    \begin{myprotocol}
        \TITLE{Client-role}{used by client $\Client$ to request transaction $\Transaction$}
        \STATE Send $\SignMessage{\Transaction}{\Client}$ to the primary $\Primary$.
        \STATE Await receipt of messages $\Message{inform}{\SignMessage{\Transaction}{\Client}, v, k, r}$ from 
		$\nf$ replicas.
        \STATE Considers $\Transaction$ executed, with result $r$,  as the $k$-th transaction.\label{fig:pa:cc}
        \SPACE
        \TITLE{Primary-role}{running at the primary $\Primary$ of view $v$, $\ID{\Primary} = v \bmod \n$}
        \STATE Let view $v$ start after execution of the $k$-th transaction.
        \EVENT{$\Primary$ awaits receipt of message $\SignMessage{\Transaction}{\Client}$ from client $\Client$}
	\begin{ALC@g}
            \STATE Broadcast $\Message{propose}{\SignMessage{\Transaction}{\Client}, v, k}$ to all replicas.\label{fig:pa:propose}
            \STATE $k \GETS k + 1$.
	\end{ALC@g}
        \ENDEVENT
	\EVENT{$\Primary$ receives $\nf$ message $\Message{support}{\SignShare{h}{i}, v, k}$ such that:
	    \begin{enumerate}[nosep]
                \item each message was sent by a distinct replica, $i \in \{1,\dots, n\}$; and
		\item All $\SignShare{h}{i}$ in this set can be combined to generate signature $\SignMessage{h}{}$.
            \end{enumerate}
	}
	\begin{ALC@g}
            \STATE Broadcast $\Message{certify}{\SignMessage{h}, v, k}$ to all replicas.\label{fig:pa:_certify}
	\end{ALC@g}
        \ENDEVENT
        \SPACE
        \TITLE{Backup-role}{running at every $i$-th replica $\Replica$.}\label{fig:pa_backup}
        \EVENT{$\Replica$ receives message $m := \Message{propose}{\SignMessage{\Transaction}{\Client}, v, k}$ such that:\label{fig:k-proposal}
            \begin{enumerate}[nosep]
                \item $v$ is the current view;
                \item $m$ is sent by the primary of $v$; and
                \item $\Replica$ did not accept a $k$-th proposal in $v$
            \end{enumerate}
        }\label{fig:pa:sup}
	\begin{ALC@g}
	    \STATE Compute {\changed $h := \Hash{\SignMessage{\Transaction}{\Client} || v || k}$}.
	    \STATE Compute signature share $\SignShare{h}{i}$.
            \STATE Transmit $\Message{support}{\SignShare{h}{i}, v, k}$ to $\Primary$.\label{fig:pa:support}
	\end{ALC@g}
        \ENDEVENT
        \EVENT{$\Replica$ receives messages $\Message{certify}{\SignMessage{h}, v, k}$ from $\Primary$ such that:
            \begin{enumerate}[nosep]
		\item $\Replica$ transmitted $\Message{support}{\SignShare{h}{i}, v, k}$ to $\Primary$; and
		\item $\SignMessage{h}{}$ is a valid threshold signature
            \end{enumerate}
        }\label{fig:pa:vc}
	\begin{ALC@g}
            \STATE View-commit $\Transaction$, the $k$-th transaction of $v$ ({\changed $\ViewCommitted{\Replica}{\SignMessage{\Transaction}{\Client}}{v}{k}$}).
	\end{ALC@g}
        \ENDEVENT
        \EVENT{$\Replica$ logged {\changed $\ViewCommitted{\Replica}{\SignMessage{\Transaction}{\Client}}{v}{k}$} and\\
                \qquad\qquad has logged $\Executed{\Replica}{t'}{v'}{k'}$ for all $0 \leq k' < k$}\label{fig:pa:exec}
	\begin{ALC@g}
            \STATE Execute $\Transaction$ as the $k$-th transaction of $v$ ({\changed $\Executed{\Replica}{\SignMessage{\Transaction}{\Client}}{v}{k}$}).
            \STATE Let $r$ be the result of execution of $\Transaction$ (if there is any result).
            \STATE Send $\Message{inform}{\Hash{\SignMessage{\Transaction}{\Client}}, v, k, r}$ to $\Client$.\label{fig:pa:inform}
	\end{ALC@g}
        \ENDEVENT
    \end{myprotocol}
    \caption{The normal-case algorithm of \PoE{}.}\label{fig:pa}
\end{figure}

To prove correctness of \PoE{} in all cases, we will need the following technical safety-related property of view-commits.

\begin{proposition}\label{prop:non_divergent}
Let $\Replica_i$, $i \in \{1, 2\}$,  be two non-faulty replicas that view-committed to $\SignMessage{\Transaction_i}{\Client_i}$ as the $k$-th transaction of view $v$ ({\changed $\ViewCommitted{\Replica}{\SignMessage{\Transaction}{\Client}}{v}{k}$}). If $\n > 3\f$, then $\SignMessage{\Transaction_1}{\Client_1} = \SignMessage{\Transaction_2}{\Client_2}$.
\end{proposition}
\begin{proof}
{\changed
Replica $\Replica_i$ only view-committed to $\SignMessage{\Transaction_i}{\Client_i}$ after $\Replica_i$ received 
$\Message{certify}{{\SignMessage{h}{}}, v, k}$ from the primary $\Primary$ (Line~\ref{fig:pa:vc} of Figure~\ref{fig:pa}). 
This message includes a threshold signature $\SignMessage{h}{}$, whose construction requires signature shares from a set $S_i$ of $\nf$ distinct replicas. 
Let $X_i = S_i \difference \Faulty$ be the non-faulty replicas in $S_i$. As $\abs{S_i} = \nf$ and $\abs{\Faulty} = \f$, we have $\abs{X_i} \geq \nf - \f$. The non-faulty replicas in $X_i$ will only send a single \MName{support} message for the $k$-th transaction in view $v$ (Line~\ref{fig:pa:sup} of Figure~\ref{fig:pa}). Hence, if $\SignMessage{\Transaction_1}{\Client_1} \neq \SignMessage{\Transaction_2}{\Client_2}$, then $X_1$ and $X_2$ must not overlap and $\nf \geq \abs{X_1 \union X_2} \geq 2(\nf - \f)$ must hold. As $\n = \nf + \f$,  this simplifies to $3\f \geq \n$, which contradicts $\n > 3\f$. Hence, we conclude $\SignMessage{\Transaction_1}{\Client_1} = \SignMessage{\Transaction_2}{\Client_2}$.}
\end{proof}

We will later use Proposition~\ref{prop:non_divergent} to show that \PoE{} provides speculative non-divergence.
Next, we look at typical cases in which the normal-case of \PoE{} is interrupted:

\begin{example}\label{ex:failures}
A malicious primary can try to affect \PoE{} by not conforming to the normal-case algorithm in the following ways:
\begin{enumerate}
    \item By sending proposals for different transactions to different non-faulty replicas. 
In this case, Proposition~\ref{prop:non_divergent} guarantees that at most a single such proposed transaction will get view-committed by any non-faulty replica. 
    \item\label{ex:failures:dark} By keeping some non-faulty replicas in the dark by not sending proposals to them. In this case, the remaining non-faulty replicas can still end up view-committing the transactions as long as at least $\nf-\f$ non-faulty replicas receive proposals: the faulty replicas in $\Faulty$ can take over the role of up to $\f$ non-faulty replicas left in the dark (giving the false illusion that the non-faulty replicas in the dark are malicious).
    \item By preventing execution by not proposing a $k$-th transaction, even though transactions following the $k$-th transaction are being proposed.
\end{enumerate}
\end{example}

When the network is unreliable and messages do not get delivered (or not on time), then the behavior of a non-faulty primary can match that of the malicious primary in the above example. Indeed, failure of the normal-case of \PoE{} has only two possible causes: primary failure and unreliable communication. If communication is unreliable, then there is no way to guarantee continuous service~\cite{flp}. Hence, replicas simply assume failure of the current primary if the normal-case behavior of \PoE{} is interrupted, while the design of \PoE{} guarantees that unreliable communication does not affect the correctness of \PoE{}. 

To deal with primary failure, each replica maintains a timer for each request. If this timer expires (\emph{timeout}) and it has not been able to execute the request, it assumes that the primary is malicious. To deal with such a failure, replicas will replace the primary. Next, we present the \emph{view-change algorithm} that performs primary replacement.

\subsection{The View-Change Algorithm}
\label{subsec:vc}
{\changed
If \PoE{} observes failure of the primary $\Primary$ of view $v$, then \PoE{} will elect a new primary and move to the next view, view $v + 1$, via the \emph{view-change algorithm}. The goals of the view-change are
\begin{enumerate}
    \item to assure that each request that \emph{is considered executed} by any client is preserved under all circumstances; and
    \item to assure that the replicas are able to agree on a new view whenever communication is reliable.
\end{enumerate}
As described in the previous section, a client will consider its request executed if it receives a \emph{proof-of-execution} consisting of identical \MName{inform} responses from at-least $\nf$ distinct replicas. Of these $\nf$ responses, at-most $\f$ can come from faulty replicas. Hence, a client can only consider its request executed whenever the requested transaction was executed (and view-committed) by at-least \emph{$\nf - \f \geq \f+1$} non-faulty replicas in the system. We note the similarity with the view-change algorithm of \pbft{}, which will preserve any request that is \emph{prepared} by at-least $\nf -\f \geq \f+1$ non-faulty replicas.}

The view-change algorithm of \PoE{} consists of three steps. 
First, failure of the current primary $\Primary$ needs to be detected by all non-faulty replicas. 
Second, all replicas exchange information to establish which transactions were included in view $v$ and 
which were not. 
Third, the new primary $\Primary'$ proposes a new view. 
This new view proposal contains a list of the transactions executed in the previous views 
(based on the information exchanged earlier). 
Finally, if the new view proposal is valid, then replicas switch to this view; otherwise, 
replicas detect failure of $\Primary'$ and initiate a view-change for the next view ($v+2$). 
The communication of the view-change algorithm of \PoE{} is sketched in Figure~\ref{fig:vcs} and the full pseudo-code of the algorithm can be found in Figure~\ref{fig:vca}. Next, we discuss each step in detail.

\subsubsection{Failure Detection and View-Change Requests}
If a replica $\Replica$ detects failure of the primary of view $v$, then it halts the normal-case algorithm of \PoE{} for view $v$ and informs all other replicas of this failure by requesting a view-change. 
The replica $\Replica$ does so by broadcasting a message $\Message{vc-request}{v, E}$, 
in which $E$ is a summary of all transactions executed by $\Replica$ (Figure~\ref{fig:vca}, Line~\ref{fig:vca:detect}). Each replica $\Replica$ can detect the failure of primary in two ways:
\begin{enumerate}
\item $\Replica$ \emph{timeouts} while expecting normal-case operations toward executing a client request. E.g., when $\Replica$ forwards a client request to the current primary, and the current primary fails to propose this request on time.
\item $\Replica$ receives \MName{vc-request} messages, indicating that the primary of view $v$ 
failed, from $\f+1$ distinct replicas. 
As at most $\f$ of these messages can come from faulty replicas, 
at least one non-faulty replica must have detected a failure. 
In this case, $\Replica$ joins the view-change (Figure~\ref{fig:vca}, Line~\ref{fig:vca:join}).
\end{enumerate}

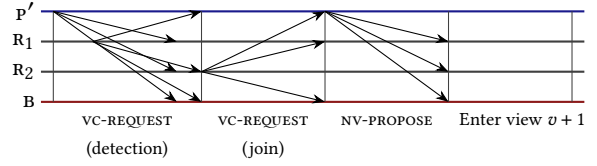
\begin{figure}[t!]
    \centering
    \begin{tikzpicture}[yscale=0.4,xscale=0.65]
        \draw[thick,draw=black!75] (0.75,   0) edge[red!50!black!90] ++(11, 0)
                                   (0.75,   1) edge ++(11, 0)
                                   (0.75,   2) edge ++(11, 0)
                                   (0.75,   3) edge[blue!50!black!90] ++(11, 0);

        \draw[thin,draw=black!75] (1,   0) edge ++(0, 3)
                                  (4,   0) edge ++(0, 3)
                                  (6.5, 0) edge ++(0, 3)
                                  (9,   0) edge ++(0, 3)
                                  (11.5,   0) edge ++(0, 3);

        \node[left] at (0.8, 0) {$\Replica[b]$};
        \node[left] at (0.8, 1) {$\Replica_2$};
        \node[left] at (0.8, 2) {$\Replica_1$};
        \node[left] at (0.8, 3) {$\Primary'$};

        \path[->] (1, 3) edge (3.5, 1)
                         edge (3.5, 2)
                         edge (3.5, 0)
                  (1.8, 2) edge (4, 3)
                           edge (4, 1)
                           edge (4, 0)
                  (4, 1) edge (6.5, 3)
                         edge (6.5, 2)
                         edge (6.5, 0)

                  (6.5, 3) edge (9, 0)
                           edge (9, 1)
                           edge (9, 2)
                           ;

        \node[below,align=center] at (2.5, 0) {\footnotesize\strut\MName{vc-request}\\\footnotesize\strut(detection)};
        \node[below,align=center] at (5.25, 0) {\footnotesize\strut\MName{vc-request}\\\footnotesize\strut(join)};
        \node[below] at (7.75, 0) {\footnotesize \strut\MName{nv-propose}};
        \node[below] at (10.5, 0) {\footnotesize \strut Enter view $v+1$};
    \end{tikzpicture}
    \caption{The current primary $\Replica[b]$ of view $v$ is faulty and needs to be replaced. The next primary, $\Primary'$, and the replica $\Replica_1$ detected this failure first and request view-change via  \MName{vc-request} messages. The replica $\Replica_2$ joins these requests.}\label{fig:vcs}
\end{figure}

\begin{figure}[t!]
    \begin{myprotocol}
        \TITLE{vc-request}{used by replica $\Replica$ to request view-change}
        \EVENT{$\Replica$ detects failure of the primary}\label{fig:vca:detect}
        	\begin{ALC@g}

	    \STATE $\Replica$ halts the normal-case algorithm of Figure~\ref{fig:pa} for view $v$.
	    \STATE $E \GETS \{ (\Message{certify}{\SignMessage{h}, w, k}, \SignMessage{\Transaction}{\Client}) \mid{}$\\
                \qquad\qquad$w \leq v\text{ and }\Executed{\Replica}{\SignMessage{\Transaction}{\Client}}{w}{k}\text{ and }h = \Hash{\SignMessage{\Transaction}{\Client} || w || k}\}$.
            \STATE Broadcast $\Message{vc-request}{v, E}$ to all replicas.
            \end{ALC@g}
        \ENDEVENT
        \EVENT{$\Replica$ receives $\f+1$ messages $\Message{vc-request}{v_i, E_i}$ such that
            \begin{enumerate}[nosep]
                \item each message was sent by a distinct replica; and
                \item $v_i$, $1 \leq i \leq \f+1$, is the current view
            \end{enumerate}
        }\label{fig:vca:join}
        	\begin{ALC@g}
            \STATE $\Replica$ detects failure of the primary (join).
            \end{ALC@g}
        \ENDEVENT
	\SPACE
	\TITLE{On receiving nv-propose}{use by replica $\Replica$}
	\EVENT{$\Replica$ receives $m = \Message{nv-propose}{v + 1, m_1, m_2,..., m_{\nf}}$}\label{fig:vca:nv_ack}
	\begin{ALC@g}
            \IF{$m$ is a valid new-view proposal (similar to creating $\MName{nv-propose}$)}
            \STATE Derive the transactions $N$ for the new-view from $m_1, m_2, \dots, m_{\nf}$.
            \STATE Rollback any executed transactions not included in $N$.
            \STATE Execute the transactions in $N$ not yet executed.
            \STATE Move into view $v+1$ (see Section~\ref{sss:move_nv} for details).
            \ENDIF
            \end{ALC@g}
        \ENDEVENT
        \SPACE
        \TITLE{nv-propose}{used by replica $\Primary'$ that will act as the new primary}
        \EVENT{$\Primary'$ receives $\nf$ messages $m_i = \Message{vc-request}{v_i, E_i}$ such that
            \begin{enumerate}[nosep]
                \item these messages are sent by a set $S$, $\abs{S} = \nf$, of distinct replicas;
                \item for each $m_i$, $1\leq i \leq \nf$, sent by replica $\Replica[q]_i \in S$, 
			$E_i$ consists of a consecutive sequence of entries {\changed $(\Message{certify}{\SignMessage{h}, v, k}, \SignMessage{\Transaction}{\Client})$};
                \item $v_i$, $1 \leq i \leq \nf$, is the current view $v$; and
                \item $\Primary'$ is the next primary ($\ID{\Primary'} = (v+1) \bmod \n$)
            \end{enumerate}
        }\label{fig:vca:nvs}
        	\begin{ALC@g}
            \STATE Broadcast $\Message{nv-propose}{v + 1, m_1, m_2, ..., m_{\nf}}$ to all replicas.
            \end{ALC@g}
        \ENDEVENT
    \end{myprotocol}
    \caption{The view-change algorithm of \PoE{}.}\label{fig:vca}
\end{figure}

\subsubsection{Proposing the New View} 
To start view $v+1$, the new primary $\Primary'$ (with $\ID{\Primary'} = (v+1) \bmod \n$) 
needs to propose a new view by determining a valid list of requests that need to be preserved.
To do so, $\Primary'$ waits until it receives sufficient information. 
{\changed In specific, $\Primary'$ waits until it received \emph{valid} \MName{vc-request} messages from a set
$S \subseteq \Replicas$ of $\abs{S} = \nf$ distinct replicas.}

An $i$-th view-change request $m_i$ is considered valid if it includes a \emph{consecutive sequence} of pairs $(c,  \SignMessage{\Transaction}{\Client})$, where $c$ is a valid \MName{certify} message for request $\SignMessage{\Transaction}{\Client}$.  Such a set $S$ is guaranteed to exist when communication is reliable, 
as all non-faulty replicas will participate in the view-change algorithm.
{\changed The new primary collects the set $S$ of $\abs{S} = \nf$ valid \MName{vc-request} and proposes them in a
new view message \MName{nv-propose} to all replicas.}

\subsubsection{Move to the New View}\label{sss:move_nv}

After a replica $\Replica$ receives a \MName{nv-propose} message containing a new-view proposal
from the new primary $\Primary'$, $\Replica$ validates the content of this message. 
From the set of \MName{vc-request} messages in the new-view proposal, $\Replica$ chooses, for each $k$, the pair $(\Message{certify}{\SignMessage{h}, w, k}, \SignMessage{\Transaction}{\Client})$ proposed in the most-recent view $w$. Furthermore, $\Replica$ determines the total number of such requests $k_{\max}$. Then, $\Replica$ view-commits and executes all $k_{\max}$ chosen requests that happened before view $v+1$. 
Notice that replica $\Replica$ can skip execution of any transaction it already executed. If $\Replica$ executed transactions not included in the new-view proposal, then $\Replica$ needs to \emph{rollback} these transactions before it can proceed executing requests in view $v+1$.
After these steps, $\Replica$ can switch to the new view $v+1$.
In the new view, the new primary $\Primary'$ starts by proposing the $k_{\max} + 1$-th transaction.

{\changed
When moving into the new view, we see  the cost of speculative execution: some replicas can be forced to \emph{rollback execution} of transactions:

\begin{example}
Consider a system with non-faulty replica $\Replica$. When deciding the $k$-th request, communication became unreliable, due to which only $\Replica$ received a \MName{certify} message for request $\SignMessage{\Transaction}{\Client}$. Consequently, $\Replica$ speculatively executes $\Transaction$ and informs the client $\Client$. During the view-change, all other replicas---none of which have a \MName{certify} message for $\SignMessage{\Transaction}{\Client}$---provide their local state to the new primary, which proposes a new view that does not include any $k$-th request. Hence, the new primary will start its view by proposing client request $\SignMessage{\Transaction'}{\Client'}$ as the $k$-th request, which gets accepted. Consequently, $\Replica$ needs to rollback execution of $\Transaction$. Luckily, this is not an issue: the client $\Client$ only got at-most $\f+1 < \nf$ responses for request, does not yet have a proof-of-execution, and, consequently, does not consider $\Transaction$ executed.
\end{example}

In practice, rollbacks can be supported by, e.g., undoing the operations of transaction in reverse order, or by reverting to an old state. For the correct working of \PoE{}, the exact working of rollbacks is not important as long as the execution layer provides support for rollbacks.
}

\subsection{Correctness of \PoE{}}
\label{app:correctness}
First, we show that the normal-case algorithm of \PoE{} provides non-divergent speculative consensus when the primary is non-faulty and communication is reliable.
\begin{theorem}\label{thm:goodcase}
Consider a system in view $v$, in which the first $k-1$ transactions have been executed by all non-faulty replicas, in which the primary is non-faulty, and communication is reliable. If the primary received $\SignMessage{\Transaction}{\Client}$, then the primary can use the algorithm in Figure~\ref{fig:pa} to ensure that
\begin{enumerate}[nosep]
    \item there is non-divergent execution of $\Transaction$;
    \item $\Client$ considers $\Transaction$ executed as the $k$-th transaction; and
    \item $\Client$ learns the result of executing $\Transaction$ (if any),
\end{enumerate}
this independent of any malicious behavior by faulty replicas.
\end{theorem}
\begin{proof}
Each non-faulty primary would follow the algorithm of \PoE{} described in Figure~\ref{fig:pa} and 
send $\Message{propose}{\SignMessage{\Transaction}{\Client}, v, k}$ to all replicas (Line~\ref{fig:pa:propose}). 
In response, all $\nf$ non-faulty replicas will compute a signature share and send a \MName{support} message to the primary (Line~\ref{fig:pa:support}). Consequently, the primary will receive signature shares from $\nf$ replicas and will combine them to generate  a threshold signature $\SignMessage{h}{}$. 
The primary will include this signature $\SignMessage{h}{}$ in a \MName{certify} message and broadcast it to all replicas.
Each replica will successfully verify $\SignMessage{h}{}$ and will view-commit to $\Transaction$ (Line~\ref{fig:pa:vc}). 
As the first $k-1$ transactions have already been executed, every non-faulty replica will execute $\Transaction$. As all non-faulty replicas behave deterministically, execution will yield the same result $r$ (if any) across all non-faulty replicas. Hence, when the non-faulty replicas inform $\Client$, they do so by all sending identical messages $\Message{inform}{\Hash{\SignMessage{\Transaction}{\Client}}, v, k, r}$ to $\Client$ (Line~\ref{fig:pa:exec}--Line~\ref{fig:pa:inform}). 
As all $\nf$ non-faulty replicas executed $\Transaction$, we have non-divergent execution. Finally, as there are at most $\f$ faulty replicas, the faulty replicas can only forge up to $\f$ invalid \MName{inform} messages. Consequently, the client $\Client$ will only receive the message $\Message{inform}{\Hash{\SignMessage{\Transaction}{\Client}}, v, k, r}$ from at least $\nf$ distinct replicas, and will conclude that $\Transaction$ is executed yielding result $r$ (Line~\ref{fig:pa:cc}).
\end{proof}

At the core of the correctness of \PoE{}, under all conditions, is that no replica will rollback requests $\SignMessage{\Transaction}{\Client}$ for which client $\Client$ already received a proof-of-execution. We prove this next:

\begin{proposition}\label{prop:norollback}
Let  $\SignMessage{\Transaction}{\Client}$ be a request for which client $\Client$ already received a proof-of-execution showing that $\Transaction$ was executed as the $k$-th transaction of view $v$. If $\n > 3\f$, then every non-faulty replica that switches to a view $v' > v$ will preserve $\Transaction$ as the $k$-th transaction of view $v$.
\end{proposition}
\begin{proof}
{\changed Client $\Client$ considers $\SignMessage{\Transaction}{\Client}$ executed as the $k$-th transaction of view $v$ when it received 
identical \MName{inform}-messages for $\Transaction$ from a set $A$ of $\abs{A} = \nf$ distinct replicas (Figure~\ref{fig:pa}, Line~\ref{fig:pa:cc}). Let $B = A \difference \Faulty$ be the set of non-faulty replicas in $A$.

Now consider a non-faulty replica $\Replica$ that switches to view $v' > v$. Before doing so, $\Replica$ must have received a valid proposal $m = \Message{nv-propose}{v', m_1, ..., m_{\nf}}$ from the primary of view $v'$. Let $C$ be the set of $\nf$ distinct replicas that provided messages $m_1, \dots, m_{\nf}$ and let $D = C \difference \Faulty$ be the set of non-faulty replicas in $C$.  We have $\abs{B} \geq \nf - \f$ and $\abs{D} \geq \nf - \f$. Hence, using a contradiction argument similar to the one in the proof of Proposition~\ref{prop:non_divergent}, we conclude that there must exists a non-faulty replica $\Replica[q] \in (B \intersect D)$ that executed $\SignMessage{\Transaction}{\Client}$, informed $\Client$, and requested a view-change. 

To complete the proof, we need to show that $\SignMessage{\Transaction}{\Client}$ was proposed and executed in the last view that proposed and view-committed a $k$-th transaction and, hence, that $\Replica[q]$ will include $\SignMessage{\Transaction}{\Client}$ in its \MName{vc-request} message for view $v'$. We do so by induction on the difference $v' - v$. As the base case, we have $v' - v = 1$, in which case no view after $v$ exists yet and, hence, $\SignMessage{\Transaction}{\Client}$ must be the newest $k$-th transaction available to $\Replica[q]$. As the induction hypothesis, we assume that all non-faulty replicas will preserve $\Transaction$ when entering a new view $w$, $v < w \leq w'$. Hence, non-faulty replicas participating in view $w$ will not support any $k$-th transactions proposed in view $w$. Consequently, no \MName{certify} messages can be constructed for any $k$-th transaction in view $w$. Hence, the new-view proposal for $w' + 1$ will include $\SignMessage{\Transaction}{\Client}$, completing the proof.}
\end{proof}

As a direct consequence of the above, we have
\begin{corollary}[Safety of \PoE{}]
\PoE{} provides speculative non-divergence if $\n > 3\f$.
\end{corollary}

{\changed We notice that the view-change algorithm does not deal with minor malicious behavior (e.g., a single replica left in the dark). Furthermore, the presented view-change algorithm will recover all transactions since the start of the system, which will result in unreasonable large messages when many transactions have already been proposed. In practice, both these issues can be resolved by regularly making \emph{checkpoints} (e.g., after every 100 requests) and only including requests since the last checkpoint in each \MName{vc-request} message. To do so, \PoE{} uses a standard fully-decentralized \pbft{}-style checkpoint algorithm that enables the independent checkpointing and recovery of any request that is executed by at least $\f+1$ non-faulty replicas whenever communication is reliable~\cite{pbftj}. Finally, utilizing the view-change algorithm and checkpoints, we prove}

\begin{theorem}[Liveness of \PoE{}]
\PoE{} provides termination in periods of reliable bounded-delay communication if $\n > 3\f$.
\end{theorem}
\begin{proof}
When the primary is non-faulty, Theorem~\ref{thm:goodcase} guarantees termination as replicas continuously accept and execute requests.  If the primary is Byzantine and fails to guarantee termination for at most $\f$ non-faulty replicas, then the checkpoint algorithm will assure termination of these non-faulty replicas. Finally, if the primary is Byzantine and fails to guarantee termination for at least $\f+1$ non-faulty replicas, then it will be replaced using the view-change algorithm. For the view-change process, each replica will start with a timeout $\delta$ after it receives $\nf$ matching \MName{vc-requests} and double this timeout after each view-change (exponential backoff). When communication becomes reliable, this mechanism guarantees that all replicas will eventually view-change to the same view at the same time. After this point, a non-faulty replica will become primary in at most $\f$ view-changes, after which Theorem~\ref{thm:goodcase} guarantees termination.
\end{proof}

\subsection{Fine-Tuning and Optimizations}\label{ss:fine-optim}

To keep presentation simple, we did not include the following optimizations in the protocol description:
\begin{enumerate}
\item To reach $\nf$ signature shares, the primary can generate one itself. Hence, it only needs $\nf - 1$ shares of other replicas.
\item The \MName{propose}, \MName{support}, \MName{inform}, and \MName{nv-propose} messages are not forwarded and only need \Name{MAC}s to provide message authentication. The \MName{certify} messages need not be signed, as tampering them would invalidate the threshold signature. The \MName{vc-request} messages need to be signed, as they need to be forwarded without tampering.
\end{enumerate}
{\changed Finally, the design of \PoE{} is fully compatible with \emph{out-of-order processing} as a replica only supports proposals for a $k$-th transaction if it had not previously supported another $k$-th proposal (Figure~\ref{fig:pa}, Line~\ref{fig:k-proposal}) and only executes a $k$-th transaction if it has already executed all the preceding transactions (Figure~\ref{fig:pa}, Line~\ref{fig:pa:exec}). As the size of the active out-of-order processing window determines how many client requests are being processed at the same time (without receiving a proof-of-execution), the size of the active window determines the number of transactions that can be rolled back during view-changes.}

%% file: prelim.tex
\subsection{System model and notations}\label{ss:poe_model}
Before providing a full description of our \PoE{} protocol, we present the system model we use and the relevant notations.

A system is a set $\Replicas$ of \emph{replicas} that process client requests.
We assign each replica $\Replica \in \Replicas$ a unique identifier $\ID{\Replica}$ 
with $0 \leq \ID{\Replica} < \abs{\Replicas}$. 
We write $\Faulty \subseteq \Replicas$ to denote the set of \emph{Byzantine replicas} 
that can behave in arbitrary, possibly coordinated and malicious, manners. 
We assume that non-faulty replicas (those in $\Replicas \difference \Faulty$) behave in accordance to the protocol and are deterministic: 
on identical inputs, all non-faulty replicas must produce identical outputs. 
We do not make any assumptions on clients: all client
can be malicious without affecting \PoE{}. We write $\n = \abs{\Replicas}$, $\f = \abs{\Faulty}$, and $\nf = \abs{\Replicas \difference \Faulty}$ 
to denote the number of replicas, faulty replicas, and non-faulty replicas, respectively. 
We assume that $\n > 3\f$ ($\nf > 2\f$).

We assume \emph{authenticated communication}: 
Byzantine replicas are able to impersonate each other, 
but replicas cannot impersonate non-faulty replicas. 
Authenticated communication is a minimal requirement to deal with Byzantine behavior.
Depending on the type of message, we use message authentication codes (\Name{MAC}s) or threshold signatures (\Name{TS}s) to achieve authenticated communication~\cite{cryptobook}. \Name{MAC}s are based on symmetric cryptography in which every pair of communicating nodes has a  {\em secret key}. We expect non-faulty replicas to keep their {\em secret keys} hidden. \Name{TS}s are based on asymmetric cryptography. In specific, each replica holds a distinct {\em private key}, which it can use to create a signature share. Next, one can produce a valid threshold signature given at least $\nf$ such signature shares (from distinct replicas). We write $\SignShare{v}{i}$ to denote the signature share of the $i$-th replica for signing value $v$. Anyone that receives a set $T = \{\SignShare{v}{j} \mid j \in T' \}$ of signature shares for $v$ from $\abs{T'} = \nf$ distinct replicas, can aggregate $T$ into a single signature $\SignMessage{v}{}$. This digital signature can then be verified using a public key. 

We also employ a \emph{collision-resistant cryptographic hash function} $\Hash{\cdot}$ that can map an arbitrary value $v$ to a constant-sized digest $\Hash{v}$~\cite{cryptobook}. We assume that it is practically impossible to find another value $v'$, $v \neq v'$, such that $\Hash{v} = \Hash{v'}$. 
We use notation $v || w$ to denotes the \emph{concatenation} of two values $v$ and $w$.

Next, we define the consensus provided by \PoE{}.
\begin{definition}\label{def:consensus}
A single run of any \emph{consensus protocol} should satisfy the following requirements: 
\begin{description}
    \item[Termination.] Each non-faulty replica executes a transaction.
    \item[Non-divergence.] All non-faulty replicas execute the same transaction.
\end{description}
Termination is typically referred to as \emph{liveness}, whereas non-divergence is 
typically referred to as \emph{safety}. In \PoE{}, execution is speculative: replicas can execute and rollback transactions. To provide safety, \PoE{} provides speculative non-divergence instead of non-divergence:
\begin{description}
    \item[Speculative non-divergence.] If $\nf - \f\geq \f+1$ non-faulty replicas accept and execute the same
transaction $\Transaction$, then all non-faulty replicas will eventually accept and execute $\Transaction$ (after rolling back any other executed transactions).
\end{description}
\end{definition}

To provide \emph{safety}, we do not need any other assumptions on communication or on clients. 
Due to well-known impossibility results for asynchronous consensus~\cite{flp}, we can only provide \emph{liveness} in periods of \emph{reliable bounded-delay communication} 
during which all messages sent by non-faulty replicas will arrive at their destination within some maximum delay.

%% file: appendix.tex
\subsection{Designing \PoE{} using MACs}
\label{app:mac}
The design of \PoE{} can be adapted to only use message authentication codes (\Name{MAC}s) to authenticate communication. This will sharply reduce the computational complexity of \PoE{} and eliminate one round of communication, this at the cost of higher \emph{quadratic} overall communication costs (see Figure~\ref{fig:poe}).

The usage of only \Name{MAC}s makes it impossible to obtain threshold signatures or reliably forward messages (as forwarding replicas can tamper with the content of unsigned messages). Hence, using \Name{MAC}s requires changes to how client requests are included in proposals (as client requests are forwarded), to the normal-case algorithm of \PoE{} (which uses threshold signatures), and to the view-change algorithm of \PoE{} (which forwards \MName{vc-request} messages). The changes to the proposal of client requests and to the view-change algorithm can be derived from the strategies used by \pbft{} to support \Name{MAC}s~\cite{pbftj}. Hence, next we only review the changes to the normal-case algorithm of \PoE{}.

Consider a replica $\Replica$ that receives a \MName{propose} message from the primary $\Primary$. Next, $\Replica$ needs to determine whether at least $\nf$ other replicas received the same proposal, which is required to achieve speculative non-divergence (see Proposition~\ref{prop:non_divergent}). When using \Name{MAC}s, $\Replica$ can do so by replacing the all-to-one support and one-to-all certify phases by a single all-to-all \emph{support phase}. In the support phase, each replica agrees to \emph{support} the first proposal  $\Message{propose}{\SignMessage{\Transaction}{\Client}, v, k}$ it receives from the primary by broadcasting a message $\Message{support}{\Hash{\SignMessage{\Transaction}{\Client}}, v, k}$ to all replicas.  After this broadcast, each replica waits until it receives \MName{support} messages, identical to the  message it sent, from $\nf$ distinct replicas. If $\Replica$ receives these messages, it \emph{view-commits} to $\Transaction$  as the $k$-th transaction in view $v$ and schedules $\Transaction$ for execution. We have sketched this algorithm in  Figure~\ref{fig:poe}.

%% file: impl.tex
\section{\ExpoDB{} Fabric}  
\label{s:impl}
To test our design principles in practical settings, we implement our \PoE{} protocol 
in our \ExpoDB{} fabric~\cite{bc-processing,geobft,resilientdb,rcc,bftbook}.
\ExpoDB{} provides its users access to a state-of-the-art replicated transactional engine and fulfills the need 
of a high-throughput permissioned blockchain fabric.
\ExpoDB{} helps us to realize the following goals:
(i) implement and test different consensus protocols;
(ii) balance the tasks done by a replica through a {\em parallel pipelined architecture};
(iii) minimize the cost of communication through {\em batching} client transactions; and
(iv) enable use of a secure and efficient ledger. 
Next, we present a brief overview of our \ExpoDB{} fabric.

\ExpoDB{} lays down a {\em client-server} architecture 
where clients send their transactions to servers for processing. 
We use Figure~\ref{fig:pipeline-replica} to
illustrate the multi-threaded pipelined architecture associated with each replica.
At each replica, we spawn multiple {\em input} 
and {\em output} threads for communicating with the network.

{\bf Batching.} 
During our formal description of \PoE{}, we assumed that the \MName{propose} 
message from the primary includes a single client request. 
An effective way to reduce the overall cost of consensus is by aggregating several client requests in a single batch and use one consensus step to reach agreement on all these requests~\cite{pbftj,zyzzyva,sbft}.
To maximize performance, \ExpoDB{} facilitates batching requests at both replicas and clients.
 
At the primary replica, we spawn multiple {\em batch-threads} 
that aggregate clients requests into a batch.
The input-threads at the primary receive client requests, assign them a 
 sequence number and enqueue these requests in the {\em batch-queue}. 
In \ExpoDB{}, all batch-threads share a common {\em lock-free queue}.
When a client request is available, a batch-thread dequeues the request 
and continues adding it to an existing batch until the batch has reached 
a pre-defined size.
Each batching-thread also hashes the requests in a batch to create a unique digest.

All other messages received at a replica are enqueued by the input-thread in 
the {\em work-queue} to be processed by the single {\em worker-thread}. 
Once a replica receive a \MName{certify} message from the primary, it forwards the 
request to the {\em execute-thread} for execution.
Once the execution is complete, the execution-thread creates an \MName{inform} message, which 
is transmitted to the client.

\begin{figure}	
	\centering
    	\includegraphics[width=\columnwidth]{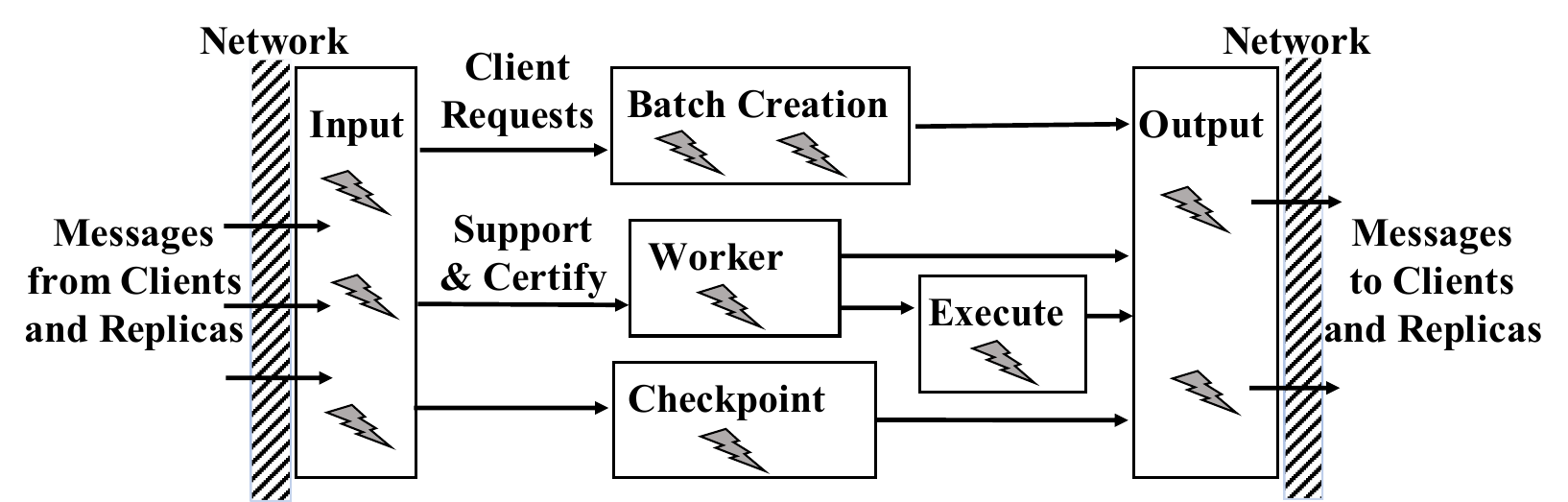}
  	\caption{\small Multi-threaded Pipelines at different replicas.}
        \label{fig:pipeline-replica}
\end{figure}

{\bf Ledger Management.} 
We now explain how we efficiently 
maintain a blockchain ledger across different replicas.
A blockchain is an immutable ledger, where blocks are 
chained as a linked-list. 
An $i$-{th} block can be represented as
$B_i \GETS \{k, d, v,\allowbreak H(B_{i-1}) \}$, in which $k$ is the sequence number of the client request, $d$ the digest of the request, $v$ the view number, and  $H(B_{i-1})$ the hash of the previous block.
In \ExpoDB{}, prior to any consensus, 
we require the {\em first} primary replica to create a 
{\em genesis block}~\cite{bc-processing}.
This genesis block acts as the first block in the blockchain and contains some basic data. We use the hash of the identity of  the initial primary, as this information is available to each participating replicas (eliminating the need for any extra communication to exchange this block).

After the genesis block, each replica can independently 
create the next block in the blockchain.
As stated above, each block corresponds to some batch of transactions.
A block is only created by the execute-thread once
it completes executing a batch of transactions. 
To create a block, the execute-thread hashes the previous 
block in the blockchain and creates a {\em new block}.
To prove the validity of individual blocks, \ExpoDB{} stores the {\em proof-of-accepting the $k$-th request} in the $k$-th block. 
In \PoE{}, such a proof includes the threshold signature sent by the primary as part of the \MName{certify} message.

%% file: eval.tex
\section{Evaluation}  
\label{s:eval}
We now analyze our design principles in practice.
To do so, we evaluate our \PoE{} protocol against four state-of-the-art \BFT{} protocols.
There are many \BFT{} protocols we could compare with. Hence, we pick a representative sample: (1) \ZZ---as it has the absolute minimal cost in the fault-free case, (2) \pbft{}---as it is a common baseline (the used design is based on BFTSmart~\cite{bftsmart}), (3) \SBFT{}---as it is a safer variation of \ZZ{}, and (3) \hotstuff{}---as it is a linear-communication protocol that adopts the notion of rotating leaders.
Through our experiments, we want to answer the following questions:
\begin{enumerate}[label=(Q\arabic*),ref={Q\arabic*}i,nosep]
\item\label{eq:1fail} How does \PoE{} fare in comparison with the other protocols under failures?
\item\label{eq:batching} Does \PoE{} benefits from batching client requests?
\item\label{eq:zero-payload} How does \PoE{} perform under zero payload?
\item\label{eq:scale} How scalable is \PoE{} on increasing the number of replicas participating in 
the consensus, in the normal-case?

\end{enumerate}

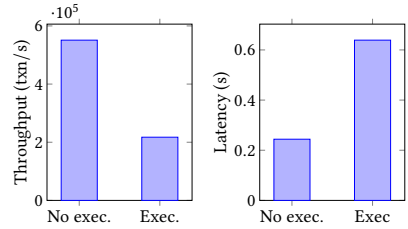
\begin{figure}[t]
	
	\centering
        \begin{tikzpicture}[plot]
            \begin{axis}[barstyle,
                         ylabel={\smash{Throughput (\si{\text{txn}\per\second})}},
                         symbolic x coords={ne,e},
                         xticklabels={No exec.,Exec.}]
                \addplot table[x={method},y={measure}] {\dataTPUTfnEXEC};
            \end{axis}
        \end{tikzpicture}
        \quad
        \begin{tikzpicture}[plot]
            \begin{axis}[barstyle,
                         ylabel={\smash{Latency (\si{\second})}},
                         symbolic x coords={ne,e},
                         xticklabels={No exec.,Exec}]
                \addplot table[x={method},y={measure}] {\dataLATfnEXEC};
            \end{axis}
        \end{tikzpicture}
	\caption{Upper bound on performance when primary only replies to clients (\emph{No exec.}) and when primary executes a request and replies to clients (\emph{Exec.}).}
	\label{fig:upper-tput}
	\label{fig:upper-lat}
	
\end{figure}

{\bf \em Setup.}\label{ss:setup}
We run our experiments on the Google Cloud, and
deploy each replicas on a $c2$ machine having a $16$-core Intel Xeon Cascade Lake CPU running at $\SI{3.8}{\giga\hertz}$ with $\SI{32}{\giga\byte}$ memory.
We deploy up to $\SI{320}{\kilo{}}$ clients on $16$ machines. 
\revised{To collect results after reaching a steady-state,} we run each experiment for $\SI{180}{\second}$: 
the first $\SI{60}{\second}$ are warmup, and measurement results are collected over the next $\SI{120}{\second}$. We average our results over three runs.

{\bf \em Configuration and Benchmarking.} 
For evaluating the protocols, we employed YCSB~\cite{ycsb} from Blockbench's macro benchmarks~\cite{blockbench}.
Each client request queries a YCSB table that holds half a million active records.
We require $90\%$ of the requests to be write queries \revised{as the majority of typical blockchain transactions are updates to existing records}.  
Prior to the experiments, each replica is initialized with an identical copy of the YCSB table. 
The client requests generated by YCSB follow a Zipfian distribution and are heavily skewed (skew factor $0.9$).

Unless {\em explicitly} stated, we use the following configuration for all experiments. 
We perform scaling experiments by varying replicas from $4$ to $91$.
We divide our experiments in two dimensions:
(1) {\em Zero Payload} or {\em Standard Payload}, and
(2) {\em Failures} or {\em Non-Failures}.
We employ batching with a batch size of $100$ \revised{as the percentage increase in throughput on larger batch sizes is small}. 

Under Zero Payload conditions, all replicas execute $100$ dummy instructions per batch, while the
primary sends an empty proposal (and not a batch of $100$ requests).
Under Standard Payload, with a batch size of $100$, the size of $\MName{Propose}$ message is $\SI{5400}{\byte}$, 
of $\MName{Response}$ message is $\SI{1748}{\byte}$, and of other messages is around $\SI{250}{\byte}$. 
For experiments with failures, we force one backup replica to crash. 
Additionally, we present an experiment that illustrates the effect of primary failure.
We measure {\em throughput} as transactions executed per second.
We measure {\em latency} as the time from when the client sends a request to the time when the client receives a response.

{\bf \em Other protocols:} 
\label{ss:other}
We also implement 
\pbft{}, \ZZ{}, \SBFT{} and \hotstuff{} in our \ExpoDB{} fabric. We refer to Section~\ref{sec:anal} for further details on the working of \ZZ{}, \SBFT{}, and \hotstuff{}. Our implementation of \pbft{} is based on the BFTSmart~\cite{bftsmart} framework with the added benefits of out-of-order processing, pipelining, and multi-threading. In both \pbft{} and \ZZ{}, digital signatures are used for authenticating messages sent by the clients, while \Name{MAC}s are used for other messages. Both \SBFT{} and \hotstuff{} require threshold signatures for their communication. 

\subsection{System Characterization}

We first determine the upper bounds on the performance of \ExpoDB{}.  
In Figure~\ref{fig:upper-tput}, we present the maximum throughput and latency 
of \ExpoDB{}
when there is {\em no communication} among the replicas.
We use the term {\em No Execution} to refer to the case where all  clients send their request to 
the primary replica and primary simply responds back to the client. We count every query responded back 
in the system throughput.
We use the term {\em Execution} to refer to the case where the primary replica executes each query 
before responding back to the client.

\revised{The  architecture of \ExpoDB{} (see Section~\ref{s:impl}) states the use of one worker thread. In these experiments, we maximize system performance by allowing up to two threads to work independently at the primary replica 
without ordering any queries.}
Our results indicate that the system can attain high throughputs (up to $\SI{500}{\kilo\text{txn}\per\second}$) and 
can reach low latencies (up to $\SI{0.25}{\second}$).
Notice that if we employ additional worker-threads, our \ExpoDB{} 
fabric can easily attain higher throughput. 

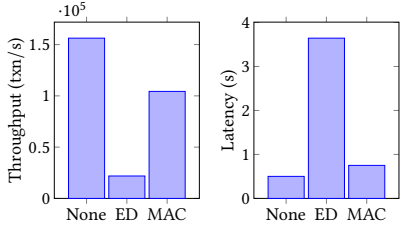
\begin{figure}[t]
	
	\centering
        \begin{tikzpicture}[plot]
            \begin{axis}[barstyle,
                         ylabel={\smash{Throughput (\si{\text{txn}\per\second})}},
                         symbolic x coords={ns,ed,cmac},
                         xticklabels={None,ED,MAC}]
                \addplot table[x={method},y={measure}] {\dataTPUTfnSIGNS};
            \end{axis}
        \end{tikzpicture}
        \quad
        \begin{tikzpicture}[plot]
            \begin{axis}[barstyle,
                         ylabel={\smash{Latency (\si{\second})}},
                         symbolic x coords={ns,ed,cmac},
                         xticklabels={None,ED,MAC}]
                \addplot table[x={method},y={measure}] {\dataLATfnSIGNS};
            \end{axis}
        \end{tikzpicture}
	\caption{System performance using three different signature schemes. In all cases, $\n = 16$ replicas participate in consensus.}
	\label{fig:crypto-lat}
	\label{fig:crypto-tput}
	
\end{figure}

\subsection{Effect of Cryptographic Signatures.}
\label{app:crypto}
\ExpoDB{} enables a flexible design where replicas and clients can employ 
both digital signatures (threshold signatures) and message authentication codes.
This helps us to implement \PoE{} and other consensus protocols in \ExpoDB{}.

To achieve authenticated communication using symmetric cryptography, 
we employ a combination of CMAC and AES~\cite{cryptobook}.
Further, we employ ED25519-based digital signatures to enable asymmetric cryptographic signing.
For generating efficient threshold signature scheme, we use Boneh--Lynn--Shacham (BLS) signatures~\cite{cryptobook}.
To create message digests and for hashing purposes, we use the SHA256 algorithm.

Next, we determine the cost of different cryptographic signing schemes.
For this purpose, we run three different experiments in which
(i) no signature scheme is used (\emph{None});
(ii) everyone uses digital signatures based on ED25519 (\emph{ED}); and
(iii) all replicas use CMAC+AES for signing, while clients sign their message using ED25519 (\emph{MAC}).
In these three experiments, we run \pbft{} consensus among $16$ replicas.
In Figure~\ref{fig:crypto-tput}, we illustrate the 
throughput attained and latency incurred by \ExpoDB{} for the experiments.
Clearly, the system attains its highest throughput when no signatures are employed. 
However, such a system cannot handle malicious attacks.
Further, using just digital signatures for signing messages can prove to be expensive.
An optimal configuration can require clients to sign their messages using digital signatures, 
while replicas can communicate using \Name{MAC}s.

\begin{figure*}
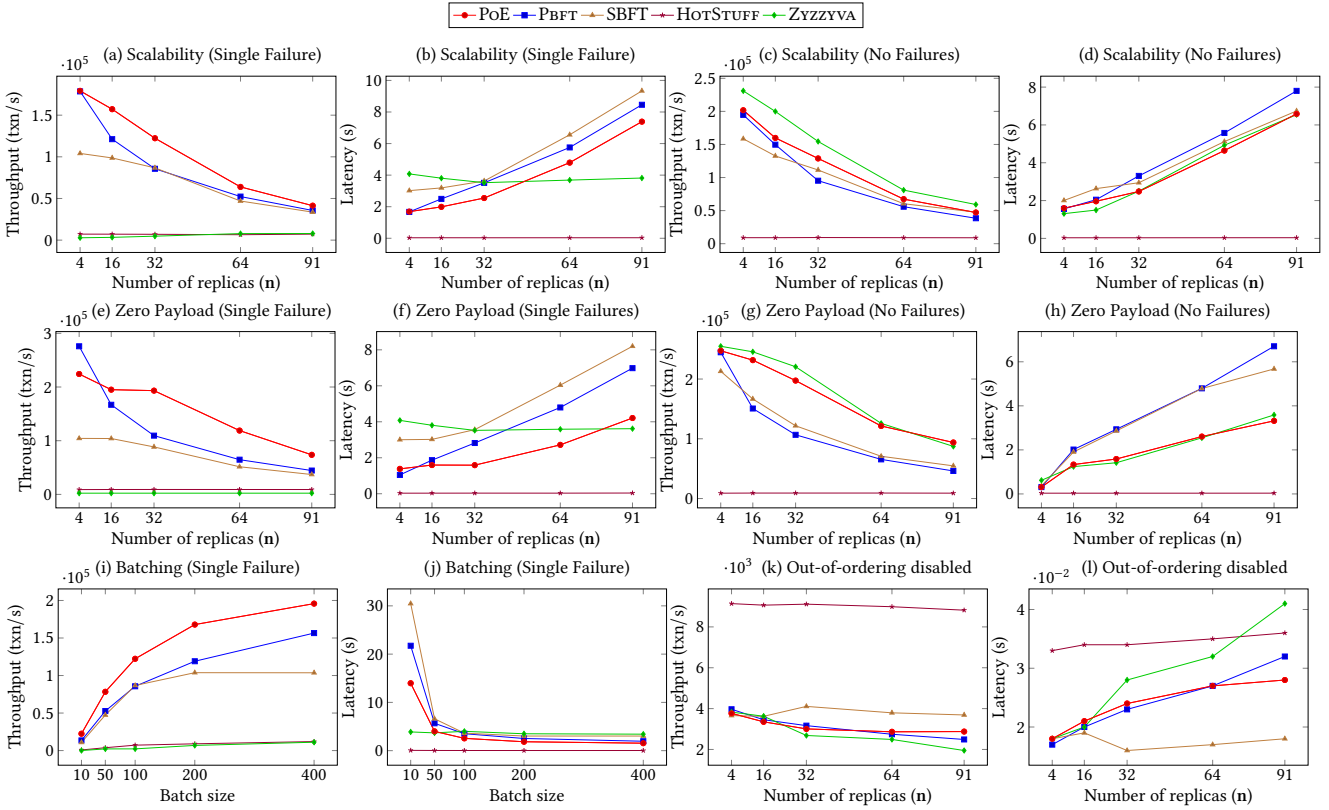


\centering
\scalebox{0.5}{\ref{mainlegend}}\\[5pt]
\resultgraph{\dataTputNodes}{\hspace{9mm}(a) Scalability (Single Failure)}{\axisnodes}{\axistput}{\axisticksnodes}%
\resultgraph{\dataLatNodes}{(b) Scalability (Single Failure)}{\axisnodes}{\axislat}{\axisticksnodes}%
\resultgraph{\dataTputNodesFF}{(c) Scalability (No Failures)}{\axisnodes}{\axistput}{\axisticksnodes}%
\resultgraph{\dataLatNodesFF}{(d) Scalability (No Failures)}{\axisnodes}{\axislat}{\axisticksnodes}\\
\resultgraph{\dataTputNodesZeroPL}{\hspace{9mm}(e) Zero Payload (Single Failure)}{\axisnodes}{\axistput}{\axisticksnodes}%
\resultgraph{\dataLatNodesZeroPL}{(f) Zero Payload (Single Failures)}{\axisnodes}{\axislat}{\axisticksnodes}%
\resultgraph{\dataTputNodesZeroPLFF}{\hspace{9mm}(g) Zero Payload (No Failures)}{\axisnodes}{\axistput}{\axisticksnodes}%
\resultgraph{\dataLatNodesZeroPLFF}{(h) Zero Payload (No Failures)}{\axisnodes}{\axislat}{\axisticksnodes}\\
\resultgraph{\dataTputBatch}{(i) Batching (Single Failure)}{\axisbatches}{\axistput}{\axisticksbatches}%
\resultgraph{\dataLatBatch}{(j) Batching (Single Failure)}{\axisbatches}{\axislat}{\axisticksbatches}%
\resultgraph{\dataTputNodesChain}{$\cdot 10^3$ \hspace{4mm}~~(k) Out-of-ordering disabled}{\axisnodes}{\axistput}{\axisticksnodes}%
\resultgraph{\dataLatNodesChain}{\hspace{9mm}(l) Out-of-ordering disabled}{\axisnodes}{\axislat}{\axisticksnodes}
\caption{Evaluating system throughput and average latency incurred by \PoE{} and other \BFT{} protocols.}
\label{fig:poe-plots}

\end{figure*}

\subsection{Scaling Replicas under Standard Payload}
\label{ss:scale-standard-payload}
In this section, we evaluate scalability of \PoE{} both under backup failure and no failures.

\begin{enumerate}[nosep,wide,label=(\arabic*)]
\item {\bf \em Single Backup Failure.}
We use Figures~\ref{fig:poe-plots}(a) and~\ref{fig:poe-plots}(b) to illustrate the throughput and latency attained 
by the system on running different consensus protocols under a backup failure.
These graphs affirm our claim that \PoE{} attains higher throughput 
and incurs lower latency than all other protocols. 

In case of \pbft{}, each replica participates in two phases of quadratic communication, 
which limits its throughput. 
For the twin-path protocols such as \ZZ{} and \SBFT{}, 
a single failure is sufficient to cause massive reductions in their system throughputs.
Notice that the collector in \SBFT{} and the clients in \ZZ{} have to wait for 
messages from all $\n$ replicas, respectively.
As predicting an optimal value for timeouts is hard~\cite{aadvark,upright}, we chose a very small value for the timeout (\SI{3}{\second}) for replicas and clients.
We justify these values, as the experiments we show later in this section show that the average latency can be as large as \SI{6}{\second}.
We note that high timeouts affect \ZZ{} more than \SBFT{}. In \ZZ{}, clients are waiting for timeouts during which they stop sending requests, which empties the pipeline at the primary, starving it from new request to propose. To alleviate such issues in real-world deployments of \ZZ{}, clients need to be able to precisely predict the latency to minimize the time the clients needs to wait between requests. Unfortunately, this is hard and runs the risk of ending up in the expensive slow path of \ZZ{} whenever the predicted latency is slightly off. In \SBFT{}, the collector may timeout waiting for threshold shares for the $k$-th round
while the primary can continues propose requests for future round $l$, $l > k$.
Hence, in \SBFT{} replicas have more opportunity to occupy themselves with useful work.

\hotstuff{} attains significantly low throughput due to its sequential primary-rotation model in which each of its primaries has to wait for the previous primary before proposing the next request, which leads to a huge reduction in its throughput.
Interestingly, \hotstuff{} incurs the least average latency among all protocols.
This is a result of intensive load on the system when running other protocols. 
As these protocols process several requests concurrently (see the multi-threaded architecture in Section~\ref{s:impl}), 
these requests spend on average more time in the queue before being processed by a replica.
Notice that all out-of-order consensus protocols employ this trade off: a small sacrifice on latency yields
higher gains on system throughput. 
 
In case of \PoE{}, its high throughputs under failures is a result of its 
three-phase linear protocol that does not rely on any twin-path model.
To summarize, \PoE{} attains up to $43\%$, $72\%$, $24\times$ and $62\times$ 
more throughputs than \pbft{}, \SBFT{}, \hotstuff{} and \ZZ{}.


\item {\bf \em No Replica Failure.}
We use Figures~\ref{fig:poe-plots}(c) and~\ref{fig:poe-plots}(d) to illustrate the throughput and latency attained 
by the system on running different consensus protocols in fault-free conditions.
These plots help us to bound the maximum throughput that can be attained by different 
consensus protocols in our system.

First, as expected, in comparison to the Figures~\ref{fig:poe-plots}(a) and~\ref{fig:poe-plots}(b), 
the throughputs for \PoE{} and \pbft{} are slightly higher. 
Second, \PoE{} continues to outperform both \pbft{} and \hotstuff{}, for the reasons described earlier.
Third, both \ZZ{} and \SBFT{} are now attaining higher throughputs as their clients and collector no longer timeout, respectively.
The key reason \SBFT's gains are limited  is because \SBFT{} requires five phases and becomes computation bounded.
Although \pbft{} is quadratic, it employs $\Name{MAC}$, which are cheaper to sign and verify.

Notice that the differences in throughputs of \PoE{} and \ZZ{} are small. 
\PoE{} has $20\%$ (on $91$ replicas) to $13\%$ (on $4$ replicas) less throughputs than \ZZ{}.
An interesting observation is that on $91$ replicas, \ZZ{} incurs almost the same latency as \PoE{}, even though it has higher throughput. This happens as clients in \PoE{} have to wait for only the fastest $\nf = 61$ replies, whereas a client for \ZZ{} has to wait for replies from all replicas (even the slowest ones). 
To conclude, \PoE{} attains up to $35\%$, $27\%$ and $21\times$ more throughput 
than \pbft{}, \SBFT{} and \hotstuff{}, respectively.

\end{enumerate}

\subsection{Scaling Replicas under Zero Payload}
\label{ss:scale-zero-payload}
We now measure the performance of different protocols under zero payload.
In any \BFT{} protocol, the primary starts consensus by sending a $\MName{Propose}$ 
message that includes all transactions.
As a result, this message has the largest size and is responsible for consuming the majority of the bandwidth.
A zero payload experiment ensures that each replica executes dummy instructions. 
Hence, the primary is no longer a bottleneck.

We again run these experiments for both {\bf \em Single Failure} and {\bf \em Failure-Free} cases, and
use Figures~\ref{fig:poe-plots}(e) to~\ref{fig:poe-plots}(h) to illustrate our observations.
It is evident from these figures that zero payload experiments have helped in increasing \PoE's gains.
\PoE{} attains up to $85\%$, $62\%$ and $27\times$ more throughputs than \pbft{}, \SBFT{} and \hotstuff{}, respectively.
In fact, under failure-free conditions, the throughput attained by \PoE{} is comparable to \ZZ{}.
This is easily explained.
First, both \PoE{} and \ZZ{} are linear protocols. 
Second, although in failure-free cases \ZZ{} attains consensus in one phase, its clients need to wait 
for response from all $\n$ replicas, which gives \PoE{} an opportunity to cover the gap.
However, \SBFT{} being a linear protocol does not perform as good as its other linear counterparts.
Its throughput is impacted by the delay of five phases.

\subsection{Impact of Batching under Failures}
\label{ss:batch-fail}
Next, we study the effect of batching client requests on \BFT{} protocols~\cite{pbftj,hotstuff}.
To answer this question, we measure performance as function of the number of requests in a batch ({\em the batch-size}), which we vary between $10$ and $400$. 
For this experiment, we use a system with $32$ available replicas, of which one replica has failed.

We use Figures~\ref{fig:poe-plots}(i) and~\ref{fig:poe-plots}(j) to illustrate, for each consensus protocol, the throughput and average latency attained by the system. For each protocol, increasing the batch-size also increases throughput, while decreasing the latency. 
This happens as larger batch-sizes require fewer consensus rounds to complete the exact same set of requests, reducing the cost of ordering and executing the transactions. This not only improves throughput, but also reduces client latencies as clients receive faster responses for their requests.
Although increasing the batch-size reduces the number of consensus rounds, the large message size 
causes a proportional decrease in throughput (or increase in latency). 
This is evident from the experiments at higher batch-sizes: increasing the batch-size beyond $100$ gradually curves the throughput plots towards a limit for \PoE{}, \pbft{} and \SBFT{}. 
For example, on increasing the batch size from $100$ to $400$, \PoE{} and \pbft{} see an increase in throughput by $60\%$ and 
$80\%$, respectively, while the gap in throughput reduces from $43\%$ to $25\%$.
As in the previous experiments, \ZZ{} yields a significantly lower throughput as it cannot handle failures. 
In case of \hotstuff{}, an increase in batch size does increases its throughput but due to high scaling of the graph 
this change seems insignificant. 


\subsection{Disabling Out-of-Ordering}
\label{ss:oop}
Until now, we allowed protocols like \pbft{}, \PoE{}, \SBFT{} and \ZZ{} to process requests {\em out-of-order}. 
As a result, these protocols achieve much higher throughputs than \hotstuff{}, which is restricted
by its sequential primary-rotation model.
In Figures~\ref{fig:poe-plots}(k) and~\ref{fig:poe-plots}(l), we evaluate the performance of the protocols when 
there are no opportunities for out-of-ordering.

In this setting, we require each client to only send its request when it has accepted a response for its previous query. 
As \hotstuff{} pipelines its phases of consensus into a {\em four}-phase pipeline, so we allow it to access four client requests 
(each on a distinct subsequent replica) at any time.
As expected, \hotstuff{} performs better than all other protocols at the expense of a higher latency 
as it rotates primaries at the end of each consensus, which allows it 
to pipeline four requests.
However, notice that once out-of-ordering is disabled, throughput drops from $\SI{200}{\kilo\text{transactions}\per\second}$ 
to just under a few {\em thousand} \si{\text{transactions}\per\second}. 
Hence, from a practical standpoint, out-of-ordering is simply crucial.
Further, the difference in latency of different protocols is quite small, and 
the visible variation is a result of graph scaling while the actual numbers are in the range of $\SI{20}{\milli\second}$--$\SI{40}{\milli\second}$.

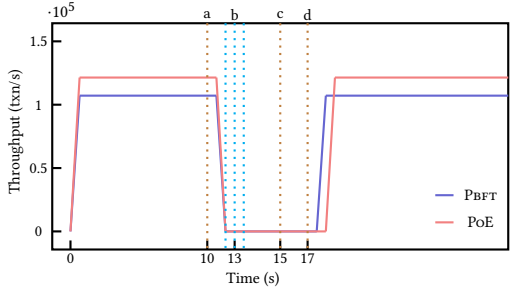
\begin{figure}
	
	\centering
	\begin{tikzpicture}[scale=1.2]
	\scriptsize
    \draw[thick] (0, 4) rectangle (5, 1.5);

	\draw[thick] (0,3.8) -- (0.1,3.8);
	\draw[thick] (0,3.1) -- (0.1,3.1);
	\draw[thick] (0,2.4) -- (0.1,2.4);
	\draw[thick] (0,1.7) -- (0.1,1.7);

	\node at (0.1,4.15) {$\cdot 10^5$};
	\node at (-0.15,3.8) {1.5};
	\node at (-0.15,3.1) {1};
	\node at (-0.15,2.4) {0.5};
	\node at (-0.15,1.7) {0};

	\node [rotate=90] at (-0.4,2.8) {Throughput (txn/s)};

	\draw[thick] (0.2,1.5) -- (0.2,1.6);
	\draw[thick] (1.7,1.5) -- (1.7,1.6);
	\draw[thick] (2.0,1.5) -- (2.0,1.6);

	\draw[thick] (2.5,1.5) -- (2.5,1.6);
	\draw[thick] (2.8,1.5) -- (2.8,1.6);


	\node at (0.2,1.4) {0};
	\node at (1.7,1.4) {10};
	\node at (2.0,1.4) {13};
	\node at (2.5,1.4) {15};
	\node at (2.8,1.4) {17};

	\node at (2.2,1.17) {Time (s)};

	\draw[thick,blue!70!black!60] (0.2,1.7) -- (0.3,3.2);
	\draw[thick,blue!70!black!60] (0.3,3.2) -- (1.8,3.2);
	\draw[thick,blue!70!black!60] (1.8,3.2) -- (1.9,1.7);
	\draw[thick,blue!70!black!60] (1.9,1.7) -- (2.9,1.7);
	\draw[thick,blue!70!black!60] (2.9,1.7) -- (3.0,3.2);
	\draw[thick,blue!70!black!60] (3.0, 3.2) -- (5,3.2);

	\draw[thick,red!90!black!50] (0.2,1.7) -- (0.3,3.4);
	\draw[thick,red!90!black!50] (0.3,3.4) -- (1.8,3.4);
	\draw[thick,red!90!black!50] (1.8,3.4) -- (1.9,1.7);
	\draw[thick,red!90!black!50] (1.9,1.7) -- (3.0,1.7);
	\draw[thick,red!90!black!50] (3.0,1.7) -- (3.1,3.4);
	\draw[thick,red!90!black!50] (3.1, 3.4) -- (5,3.4);

	\draw[thick,dotted,brown] (1.7,4) -- (1.7,1.5);
	\draw[thick,dotted,cyan] (1.9,4) -- (1.9,1.5);
	\draw[thick,dotted,cyan] (2.0,4) -- (2.0,1.5);
	\draw[thick,dotted,cyan] (2.1,4) -- (2.1,1.5);
	\draw[thick,dotted,brown] (2.5,4) -- (2.5,1.5);
	\draw[thick,dotted,brown] (2.8,4) -- (2.8,1.5);

	\node at (1.7,4.1) {a};
	\node at (2,4.1)   {b};
	\node at (2.5,4.1) {c};
	\node at (2.8,4.1) {d};
	
	\draw[thick,blue!70!black!60] (4.2,2.1) -- (4.4,2.1);
	\node at (4.7,2.1) {\pbft{}};

	\draw[thick,red!90!black!50] (4.2,1.8) -- (4.4,1.8);
	\node at (4.7,1.8) {\PoE{}};

	\end{tikzpicture}
    
    \caption{System throughput under instance failures ($\n=32$). 
	(a) replicas detect failure of primary and broadcast $\MName{vc-request}$;
	(b) replicas receives $\MName{vc-request}$ from others;
	(c) replicas receives $\MName{nv-propose}$ from new primary;
	(d) state recovery;
	}
  	\label{sfig:tput-vc}
     
\end{figure}

\subsection{Primary Failure--View Change}
In Figure~\ref{sfig:tput-vc}, we study the impact of of a benign primary failure on \PoE{} and \pbft{}. 
To recover from a primary failure, backup replicas run the view-change protocol.
We skip illustrating view-change plots for \ZZ{} and \SBFT{} as they already face severe reduction in throughput 
for a single backup failure. 
Further, \ZZ{} has an {\em unsafe} view-change algorithm and \SBFT's view-change algorithm is no less expensive than \pbft{}.
For \hotstuff{}, we do not show results as it changes primary at the end of every consensus.
Although single primary protocols face a momentary loss in throughput during view-change, these protocols 
easily cover this gap through their ability to process messages out-of-order.

For our experiments, we let 
the primary replica complete consensus for $\SI{10}{\second}$ (or around a million transactions)
and then fail. 
This causes clients to timeout while waiting for responses for their pending transactions. 
Hence, these clients forward their requests to backup replicas. 

When a backup replica receives a client request, it forwards that request to the primary and waits on a timer.
Once a replicas timeouts, it detects a primary failure and broadcasts a $\MName{vc-request}$ message to 
all other replicas---initiate view-change protocol (a).
Next, each replica waits for a new view message from the next primary.
In the meantime, a replica may receive $\MName{vc-request}$ messages from other replicas (b).
Once a replica receives $\MName{nv-propose}$ message from the new primary (c), it moves to the next view.

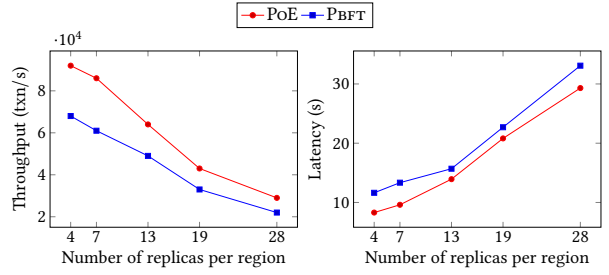
\begin{figure}

\centering
\scalebox{0.5}{\ref{wanlegend}}\\ 
\begin{tikzpicture}[plot]
    \begin{axis}[xlabel={Number of replicas per region},ylabel={\axistput},legend to name={wanlegend},legend columns=-1,width=240pt,xtick={4,7,13,19,28}]
        \addplot table[x={npr},y={POE}] {\dataTPutWAN};
        \addplot table[x={npr},y={PBFT}] {\dataTPutWAN};
        \legend{\PoE{},\pbft{}};
    \end{axis}
\end{tikzpicture}
\begin{tikzpicture}[plot]
    \begin{axis}[xlabel={Number of replicas per region},ylabel={\axislat},width=240pt,xtick={4,7,13,19,28}]
        \addplot table[x={npr},y={POE}] {\dataLatWAN};
        \addplot table[x={npr},y={PBFT}] {\dataLatWAN};
    \end{axis}
\end{tikzpicture}
\caption{{\changed System throughput and average latency incurred by \PoE{} and \pbft{} in a WAN deployment of five regions under a single failure. In the largest deployment, we have 140 replicas spread equally over these regions.}}
\label{fig:wan-plots}
\end{figure}

{\changed
\subsection{WAN Scalability}
In this section, we use Figure~\ref{fig:wan-plots} to illustrate the throughputs and latencies for different \PoE{} and \pbft{} deployments 
on a wide-area network in the presence of a single failure. 
In specific, we deploy clients and replicas across {\em five} locations across the globe: Oregon, Iowa, Montreal, the Netherlands, and 
Taiwan.
Next, we vary the number of replicas from $20$ to $140$ by equally distributing these replicas across each region.

These plots affirm our existing observations that \PoE{} outperforms existing state-of-the-art protocols and scales 
well in wide-area deployments. 
In specific, \PoE{} achieves up to $1.41\times$ higher throughput and incurs $28.67\%$ less latency than \pbft{}.
We skip presenting plots for \SBFT{}, \hotstuff{} and \ZZ{} due to their low throughputs under failures.
}

%% file: simulation.tex
\begin{figure*}
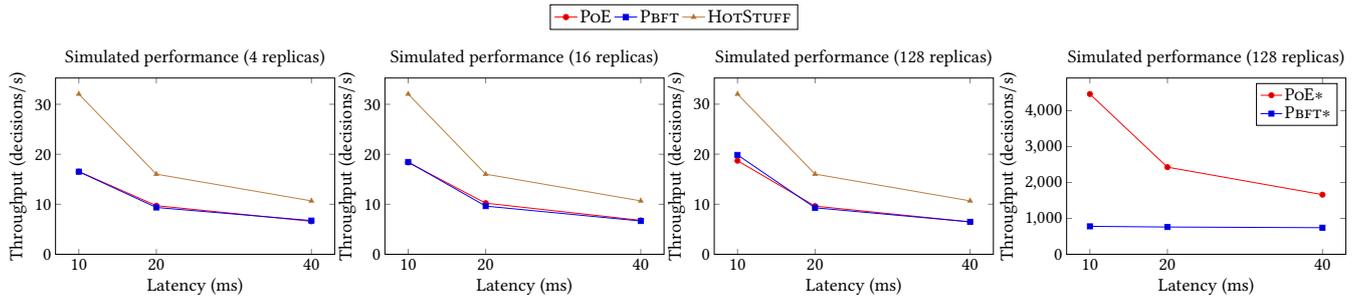

\vspace{-4mm}
\centering
\scalebox{0.5}{\ref{latlegend}}\\[5pt]
\makebox[0pt]{
\latgraph{\dataSimLatencyTPutFOUR}{Simulated performance (4 replicas)}%
\latgraph{\dataSimLatencyTPutSIXTEEN}{Simulated performance (16 replicas)}%
\latgraph{\dataSimLatencyTPutONETWOEIGHT}{Simulated performance (128 replicas)}%
\latspecialgraph{\dataSimLatencyTPutONETWOEIGHT}{Simulated performance (128 replicas)}}
\caption{The simulated number of consensus decisions  \PoE{}, \pbft{}, and \hotstuff{} can make as a function of the latency. Only the protocols in the right-most plot and marked with $*$ process requests out-of-order processing.}\label{fig:simlat}
\vspace{-4mm}
\end{figure*}

\subsection{Simulating \BFT{} Protocols}
\label{app:simulate}

To further underline that the \emph{message delay} and not \emph{bandwidth requirements} becomes a determining factor in the throughput of protocols in which the primary does not propose requests out-of-order, we performed a separate simulation of the maximum performance of \PoE{}, \pbft{}, and \hotstuff{}.
The simulation makes $500$ consensus decisions and processes all message send and receive steps, but delays the arrival of messages by a pre-determined message delay. The simulation skips any expensive computations and, hence, the simulated performance is entirely determined by the cost of message exchanges. We ran the simulation with $\n \in \{4, 16, 128\}$ replicas, for which the results can be found in Figure~\ref{fig:simlat}, first three plots.
As one can see, if bandwidth is not a limiting factor, then the performance of protocols that do not propose requests out-of-order will be determined by the number of communication rounds and the message delay. As both \pbft{} and \PoE{} have one communication round more than the two rounds of \hotstuff{}, their performance is roughly  two-thirds that of \hotstuff{}, this independent of the number of replicas or the message delay. Furthermore, doubling message delay will roughly half performance.
Finally, we also measured the maximum performance of protocols that do allow out-of-order processing of up to $250$ consensus decisions. These results can be found in Figure~\ref{fig:simlat}, last plot. As these results show, out-of-order processing increases performance by a factor of roughly $200$, even with 128 replicas.

%% file: related.tex
\section{Related Work}
\label{s:related}
Consensus is an age-old problem that received much theoretical and practical attention (see, e.g.,~\cite{paxos,raft,zab}). 
Further, the use of rollbacks is common in distributed systems. E.g., the crash-resilient replication protocol Raft~\cite{raft} allows primaries to re-write the log of any replica. In a Byzantine environment, such an approach would delegate too much power to the primary, as they can maliciously overwrite transactions that need to be preserved.


The interest in practical \BFT{} consensus protocols took off with the introduction of  \pbft{}~\cite{pbftj}.
Apart from the protocols that we already discussed, there are 
some interesting protocols that achieve efficient consensus by requiring $5\f+1$ replicas~\cite{qu-bft,hq}.
However, these protocols have been shown to work only in the cases where transactions are non-conflicting~\cite{zyzzyva}.
%
%
Some other \BFT{} protocols~\cite{less-replica-2,minbft} 
suggest the use of {\em trusted components} to 
reduce the cost of \BFT{} consensus. 
These works require only $2\f+1$ 
replicas as the trusted component 
helps to guarantee a correct ordering.
The safety of these protocols relies on the security of trusted component. 
%
In comparison, \PoE{} does (i) not require extra replicas, (ii) not depend on clients, 
(iii) not require trusted components, and 
(iv) not need the two phases of quadratic communication required by \pbft{}.

As a promising future direction, Castro~\cite{pbftj} also suggested exploring speculative optimizations 
for \pbft{}, which he referred to as tentative execution. 
However, this lacked: (i) formal description, (ii) non-divergence safety property, 
(iii) specification of rollback under attacks, (iv) re-examination of the view change protocol, and (v) any actual evaluation.

{\bf \em Consensus for Blockchains:}
Since the introduction of Bitcoin~\cite{bitcoin}, the well-known cryptocurrency that 
led to the coining of the term blockchain, several new \BFT{} consensus protocols that cater to cryptocurrencies 
have been designed~\cite{pow,ppcoin}.
Bitcoin~\cite{bitcoin} employs the {\em Proof-of-Work}~\cite{pow} consensus protocol (\POW{}), which
is computationally intensive, achieves low throughput, and can cause forks (divergence) in the blockchain: separate chains can exist on non-faulty replicas, which in turn can cause {\em double-spending attacks}~\cite{bc-processing}.
Due to these limitations, several other similar algorithms have been proposed.
%
E.g., {\em Proof-of-Stake} (PoS)~\cite{ppcoin}, which is design such that any replica owning $n\%$ of the total 
resources gets the opportunity to create $n\%$ of the 
new blocks.
As PoS is resource driven, it can face attacks where replicas are incentivized to work simultaneously  on several forks of the blokchain, without ever trying to eliminate these forks.



There are also a set of interesting alternative designs 
such as ConFlux~\cite{conflux}, 
Caper~\cite{caper} and MeshCash~\cite{meshcash} 
that suggest the use of directed acyclic graphs (DAGs) to 
store a blockchain to improve the performance of Bitcoin. 
However, these protocols either rely on \POW{} or \pbft{} for consensus.
Meta-protocols such as RCC~\cite{rcc} and RBFT~\cite{rbft} run multiple \pbft{} consensuses in parallel. 
These protocols also aim at removing dependence on the consensus led by a single primary.
A recent protocol, PoV~\cite{pov}, provides fast \BFT{} consensus in a consortium architecture. 
PoV does this by restricting the ability to propose blocks among a subset of trusted replicas.

%
%
\PoE{} does not face the limitations faced by \POW{}~\cite{pow} and
PoS~\cite{ppcoin}.
The use of DAGs~\cite{conflux,caper,meshcash}, and sharding~\cite{ahl,rapidchain} is orthogonal to the design of \PoE{}. Hence, their use with \PoE{} can reap further 
benefits.
Further, \PoE{} can be employed by meta-protocols and does not restrict consensus to any subset of replicas.

%% file: concl.tex
\section{Conclusions}  
\label{s:concl}
We present Proof-of-Execution (\PoE{}), a novel Byzantine
fault-tolerant consensus protocol that guarantees safety and liveness 
and does so in only three linear phases.
\PoE{} decouples ordering from execution by allowing 
replicas to process messages out-of-order and execute client-transactions speculatively.
Despite these properties, \PoE{} ensures that all the replicas reach a single 
unique order for all the transactions.
Further, \PoE{} guarantees that if a client observes identical results of execution 
from a majority of the replicas, then it can reliably mark its transaction committed.
Due to speculative execution, \PoE{} may require replicas to revert executed transactions, however. 
To evaluate \PoE's design, we implement it in our \ExpoDB{} fabric. {\changed Our evaluation shows that \PoE{} achieves up-to-$80\%$ higher throughputs than existing \BFT{} protocols in the presence of failures.}
